\documentclass[a4paper]{llncs}
\usepackage[utf8]{inputenc}
\usepackage[T1]{fontenc}

\usepackage[margin=1.314in]{geometry}
\usepackage{amsmath, amssymb,verbatim, bbm,amsfonts, hyperref, mathrsfs,empheq,mathtools,dsfont,physics,stmaryrd,color,nicefrac,calrsfs,complexity}
\usepackage[mathscr]{eucal}

\usepackage{lmodern}
\usepackage{ae,aecompl}
\usepackage{eso-pic}	 
\usepackage{array}
\usepackage{accentsnew}

\newcommand{\IS}{\mathrm{IS}}

\newcommand{\ISKeygen}{Keygen}

\newcommand{\ISProver}{P}
\newcommand{\ISCheck}{V}
\newcommand{\Init}{\ISKeygen}
\newcommand{\saltsize}{l_{\mathrm{salt}}}

\newcommand{\game}{\textsc{Game:}}

\newcommand{\commsize}{l_\mathrm{comm}}

\newcommand{\SKeygen}{\mathrm{Keygen}_S}

\makeatletter
\newcommand{\subalign}[1]{%
	\vcenter{%
		\Let@ \restore@math@cr \default@tag
		\baselineskip\fontdimen10 \scriptfont\tw@
		\advance\baselineskip\fontdimen12 \scriptfont\tw@
		\lineskip\thr@@\fontdimen8 \scriptfont\thr@@
		\lineskiplimit\lineskip
		\ialign{\hfil$\m@th\scriptstyle##$&$\m@th\scriptstyle{}##$\hfil\crcr
			#1\crcr
		}%
	}%
}
\makeatother

\newcommand{\cv}{\vect{c}}

\renewcommand{\ev}{\vect{e}}

\newcommand{\sv}{\vect{s}}
\newcommand{\tv}{\vect{t}}

\newcommand{\vv}{\vect{v}}

\newcommand{\xv}{\vect{x}}
\newcommand{\yv}{\vect{y}}
\newcommand{\zv}{\vect{z}}

\newcommand{\Hm}{\vect{H}}

\usepackage{amssymb,latexsym}
\usepackage{tikz}
\usepackage[boxruled]{algorithm2e}
\usepackage{pgfplots}

\pgfplotsset{compat=1.10}

\newcommand{\seed}{\mathrm{seed}}

\newcommand{\salt}{\emph{salt}}
\renewcommand{\comm}{comm}

\usepackage{lipsum}

\newcommand\wt[1]{\abs{\vect{#1}}}

\usepackage{caption}
\usepackage{subcaption}
\usepackage{color}

\newcommand*{\eqdef}{\stackrel{\text{def}}{=}}

\newcommand{\SSign}{\mathrm{Sign}_S}
\newcommand{\SVerify}{\mathrm{Ver}_S}

\newcommand{\zo}{\{0,1\}}

\newcommand{\Unif}{\leftarrow}
\renewcommand{\ss}{\ensuremath{\mathrm{S}}}

\definecolor{coll}{HTML}{000090}

\newcommand{\SD}{\mathrm{SD}}

\newcommand{\PKP}{\mathrm{PKP}}
\newcommand{\Perm}{Perm}
\newcommand{\wins}{\mathrm{ wins }}

\usepackage{wasysym}

\renewcommand{\aa}{\mathscr{A}}
\newcommand{\bb}{\mathscr{B}}

\newcommand{\hh}{\mathcal{H}}

\renewcommand{\ev}{\mathbf{e}}

\newcommand{\seedsize}{l_\mathrm{seed}}

\newcommand{\COMMENT}[1]{}

\def\01{\{0,1\}}
\def\01{\{0,1\}}

\newcommand{\F}{\mathbb{F}}
\newcommand{\N}{\mathbb{N}}

\newcommand{\Fq}{\F_q}

\usepackage{ae,aecompl}

\newcommand{\cadre}[1]
{
	\begin{tabular}{|p{\textwidth}|}
		\hline
		\vspace*{-0.3cm}
		#1 \\
		\hline
	\end{tabular}

}

\newlength{\plarg}
\newlength{\glarg}
\newcommand{\vect}[1]{\mathbf{#1}}
\renewcommand{\vec}[1]{\mathbf{#1}}

\usepackage[framemethod=TikZ]{mdframed}
\newcounter{openbox}[section]

\newenvironment{openbox}[2][]{%
    \refstepcounter{openbox}
 
    \ifstrempty{#1}%
    {\mdfsetup{%
        frametitle={%
            \tikz[baseline=(current bounding box.east),outer sep=0pt]
            \node[anchor=east,rectangle,fill=white]
            {};}
        }%
    }{\mdfsetup{%
        frametitle={%
            \tikz[baseline=(current bounding box.east),outer sep=0pt]
            \node[anchor=east,rectangle,fill=white]
            {~#1};}%
        }%
    }%
\mdfsetup{%
    innertopmargin=2pt,innerbottommargin=10pt,linecolor=black,%
    linewidth=0.5pt,topline=true,%
    frametitleaboveskip=\dimexpr-\ht\strutbox\relax%
}
 
\begin{mdframed}[]\relax}{%
\end{mdframed}}

\title{On the (In)security of optimized Stern-like signature schemes}
\author{André Chailloux\inst{1} \and Simona Etinski\inst{2} \\
{andre.chailloux@inria.fr}, {simona.etinski@cwi.nl}}
\institute{Inria de Paris COSMIQ team \and Centrum Wiskunde and Informatica Cryptology group}

\begin{document}
	\maketitle
	\begin{abstract}
		Stern's signature scheme is a historically important code-based signature scheme. A crucial optimization of this scheme is to generate pseudo-random vectors and permutation instead of random ones, and most proposals that are based on Stern's signature use this optimization. However, its security has not been properly analyzed, especially when we use deterministic commitments. In this article, we study the security of this optimization. We first show that for some parameters, there is an attack that exploits this optimization and breaks the scheme in time $O(2^{\frac{\lambda}{2}})$ while the claimed security is $\lambda$ bits. This impacts in particular the recent Quasy-cyclic Stern signature scheme~\cite{BGMS22}. Our second result shows that there is an efficient fix to this attack. By adding a string $\salt \in \zo^{2\lambda}$ to the scheme, and changing slightly how the pseudo-random strings are generated, we prove not only that our attack doesn't work but that for any attack, the scheme preserves $\lambda$ bits of security, and this fix increases the total signature size by only $2\lambda$ bits. We apply this construction to other optimizations on Stern's signature scheme, such as the use of Lee's metric or the use of hash trees, and we show how these optimizations improve the signature length of Stern's signature scheme.
	\end{abstract}
	
	\section{Introduction}
	
	As the NIST standardization process~\cite{Nist17} for post-quantum cryptography advances, the challenge of building efficient and secure post-quantum signature schemes becomes increasingly demanding. Though lattice-based schemes seem as the most promising candidates, their well-studied alternatives from code-based cryptography are gaining more attention recently. The lack of efficiency of code-based signature schemes, as their major drawback, is addressed in the work of many cryptographers, and different techniques that aim to circumvent this problem were proposed. For some of these techniques, however, it remained unclear if the schemes stayed secure after decreasing the signature size, i.e. if the security reduction still holds in that case.
	
	 In this article, we focus on the signature schemes derived from the identification protocol originally proposed by Jacques Stern ~\cite{Ste94}. Its security is based on the hardness of the syndrome decoding problem ~\cite{BMvT78}  and the Fiat-Shamir transform \cite{FS87}, which is well-analyzed with respect to classical and quantum adversaries. The main drawback of Stern's signature scheme in its original form, the same as of the other code-based schemes, is large signature size (order of a few hundreds kilobytes). Different variants of the scheme are thus proposed to mitigate this problem. Some of them, for example, reduce the signature sizes by reducing the cost of each round of the underlying identification protocol. Some versions use the cut and choose technique or even modify the underlying syndrome decoding problem (\cite{Ver97}, \cite{AGS11}, \cite{CVE11}, \cite{GRSZ14}, \cite{Beu20}, \cite{BBC+20}, \cite{BGMS22}, \cite{FJR21}, \cite{FJR22}).
	
	The goal of this article is to study several variants of optimized Stern's signature scheme and to assess their security. We first discuss here two features that appear in Stern's signature and other schemes that use the Fiat-Shamir transform.
	
	\paragraph{Random seeds for constructing pseudo-random strings.} In the original Stern's scheme, the signer must generate random vectors $\yv_i$ as well as a random permutations $\pi_i$. In the signature, some of these have to be revealed and for others, we only have the commitment to these values. Revealing all these values is very costly and makes the signature very long. The idea is to replace these random strings by pseudo-random strings. For example, if we want to construct a vector  $\yv \in  \Fq^n$, we use a function $E : \zo^{\seedsize} \rightarrow \Fq^n$, generate a random $\seed \in \zo^{\seedsize}$ and set $\yv = E(\seed)$. When one wants to reveal $\yv$, it is enough to reveal the $\seed$ which is much shorter.
	
	This idea has been used in several proposals but has never been properly proven and we don't even know to what extent this optimization remains secure. In order to maintain $\lambda$ bits of security\footnote{When we write maintain $\lambda$ bits of security, we mean that if the original scheme has $\lambda$ bits of security then the optimized scheme should also have $\lambda$ bits of security.} should we have $\seedsize = {2\lambda}$ in order to avoid any collision between the seeds, or can we take $\seedsize = {\lambda}$ which ensures that $\yv$ has $\lambda$ bits of randomness? As we will see, many designs use seeds of size $\lambda$ which seems at first sight to be enough.
	
	\paragraph{Using commitment from hash functions}
	The following textbook construction builds a commitment scheme from a hash function $\hh$: in order to commit to a string $s$, pick a random string $\rho \in \zo^{2\lambda}$ and send as a commitment $c = \hh(s,\rho) \in \zo^{2\lambda}$. In order to reveal $s$, one also reveals the random value $\rho$.  This commitment scheme has $\lambda$ bits of security against collision attacks, and also has a strong hiding property.  For example, in MQDSS~\cite{CHR+20} , we can read (where $k$ is the security parameter):
	\begin{quote}
		\emph{In particular the commitment functions now take an additional argument - a random string of length 2k. The reason for this change is that with this additional input it can be shown that the commitment function used in MQDSS is computationally hiding - a property needed to show the EU-CMA security of MQDSS.}
	\end{quote}
	
	The string they refer to is exactly the string $\rho$.
	However, in many proposer schemes, this randomization $\rho$ is omitted from the commitment scheme. Indeed, if one has to reveal a string $\rho$ of size $2\lambda$ for each opening, this can greatly increase the size of the signature scheme constructed from this commitment. We will call a commitment with $\rho$ a randomized commitment and a commitment without this $\rho$ a deterministic commitment. So what it the correct thing to do? Is it necessary to use randomized commitments or can we use deterministic commitments?

	Strikingly, each proposal that uses these $2$ ingredients doesn't give the same answers to these questions, some are quite conservative while others take the maximum amount of risk.  For instance
	\begin{itemize}
		\item \cite{BGMS22} in an optimized Stern's signature scheme, uses seeds of size $\lambda$ and deterministic commitments.
		\item \cite{BBC+20} relies on restricted syndrome decoding, uses seeds of size $2\lambda$ and deterministic commitments. 
		\item As we said,  the MQDSS signature scheme uses seeds of size $\lambda$ bits and commitments of size $2\lambda$. However, the commitments are probabilistic so our attack doesn't work but this is makes the scheme quite inefficient. 
	\end{itemize}
	
	\subsection{Contributions}
	
	Our contributions can be divided in $3$ parts:
	\begin{enumerate}
		\item First we show that one has indeed to be careful with some optimizations. In Stern's signature scheme, if one adds seeds of size $\lambda$ and deterministic commitments, then we only achieve $2^{\lambda/2}$ bits of security instead of $\lambda$ bits. We show an attack in $O(2^{\lambda/2})$  that recovers the secret key in this setting. This directly attacks the scheme presented in~\cite{BGMS22}.
		\item We present the $\salt$ + index construction for commitments in order to circumvent the attack by adding a minimal amount of randomness in the commitment scheme. We show that with this addition, Stern's signature scheme has $\lambda$ bits of security and the cost in terms of signature size is minimal (a total increase of $2\lambda$ bits to be precise).
		\item We look at other optimizations of Stern's signature scheme: for example by using Lee's metric or by adding Merkle trees, how to improve even more the size of Stern's signature scheme.
	\end{enumerate}
	
	 We now elaborate each of these contributions.
	
	\subsubsection{Showing the insecurity of some optimized Stern's signature schemes.}~\label{Section:IntroCommitments}

	Our first contribution is the following
	
	\begin{theorem}
		If we use seeds of size $\lambda$ and deterministic commitments in Stern's signature scheme, then there is an attack that recovers the secret key in time $O(2^{\lambda/2})$.
	\end{theorem}
	
	In other words, the choice of~\cite{BGMS22} for instance gives only $\lambda/2$ bits of security and not $\lambda$ as claimed. We present here the main idea behind our attack, which is quite simple. Assume you construct a pseudo-random string $\yv$ using a random$\seed$ st. $E(\seed) = \yv$ and that you commit to $\yv$ using a deterministic commitment $c = \hh(\yv)$. If you perform such a commitment many times, which happens when you perform many calls to the signing oracle, then you will arrive at a collision wp. $O(2^{\lambda/2})$ (at least each time there is a collision in the seeds) and an adversary can exploit such a collision. Indeed, if for one signature, this commitment is opened and for another signature it isn't opened, the adversary will be able to guess the committed value for this second signature which will break the scheme.
	
	\subsubsection{Mitigating the attack: the \salt + index construction}
	
	\paragraph{On multi-HVZK advantage.}
	Before we discuss the way to circumvent the above attack, let's discuss on how come the attack we present exists? Stern's signature scheme is constructed by applying the Fiat-Shamir transform to an identification scheme and this transformation is supposed to be secure. An identification scheme is required to be honest-verifier zero-knowledge and the identification scheme used for instance by~\cite{BGMS22} is HVZK. So where is the issue? Actually, if one wants to prove that a signature scheme constructed from an identification scheme is EUF-CMA, we  require the identification scheme to be multi-HVZK meaning that it is honest-verifier zero-knowledge even when seeing many transcripts arising from the identification scheme. Normally, this is not a big issue and the HVZK property will imply the multi-HVZK but this needn't be always the case. In \cite{GHHM21}, where the authors argue about the security of the Fiat-Shamir transform, they write
	
	\begin{quote}
		\emph{``In our security proofs, we will have to argue that collections of honestly generated transcripts are indistinguishable from collections of simulated ones. Since it is not always clear whether computational HVZK implies computational HVZK for multiple transcripts, we extend our definition, accordingly: In the multi-HVZK game, the adversary obtains a collection of transcripts (rather than a single one)."}
	\end{quote}
	
	Here, we are exactly in a setting where this subtlety matters. If we consider one call of Stern's optimized identification scheme, then one can construct a simulator st. any adversary that wants to distinguish the real transcript from a simulated transcript needs average time $O(2^{\lambda})$. However, there is an adversary that runs in time $O(2^{\lambda/2})$ which can recover the secret key with overwhelming probability if he is given $\Omega(2^{\lambda/2})$ independent transcripts. In order words, the $t$-HVZK property breaks for $t = 2^{\lambda/2}$. 
	
	\paragraph{The salt + index construction.}
	The above attack can be circumvented if we use a different hash function $\hh$ each time we perform a commitment. This is exactly the role of the string $\rho$ in the randomized commitment scheme. In theory, we could have a counter $i$ and use as a commitment $\hh(s,i)$ the $i^{th}$ time we perform a commitment. However, since we run the signing oracle several times, we would need to remember this counter between different signatures and that would make the scheme stateful. There ie however an idea along these lines that will work: 
	\begin{itemize}
		\item Each time you sign, pick a random string $\salt \in \zo^{2\lambda}$ that will act as a randomization string. What is important is that this string is the same for all the commitments used in the signature scheme so we will have to reveal only one such string. 
		\item \emph{Within each signature:} add a counter (which we call index) to ensure that within the same call to the signing oracle, we don't call the same function twice.
	\end{itemize}
	
	Since $\salt$ is the same throughout the signature, the whole signature is increased only by $\saltsize = 2\lambda$ bits. One can see that our attack will not work. Even if there are collisions in the seeds, the hashed values will be with different salts or different indices, and this negates our attack. Moreover, we show that the resulting scheme is mutli-HVZK, which ensures that no other attacks of this type exist.
	
	\begin{theorem}
		The optimized Stern's signature scheme with $\salt$ + index has
		$$ Adv^{t-HVZK}(\aa) = Adv^{SD}(|\aa|) + O(\frac{t^2}{2^{\saltsize}}) + O(\frac{t}{2^{\seedsize}}).$$
	\end{theorem}
	
	The way to read this theorem is the following: it is essentially as hard to break the $t$-HVZK property than to break the syndrome decoding problem (which is the hard computational problem underlying the whole scheme), as long as $\saltsize \ge t^2$ and $\seedsize \ge t$. $t$ here actually corresponds to the number of calls allowed to the signing oracle, which is often taken to be $2^\lambda$. This motivates the choice $l_\salt = 2\lambda$ and $l_\seed = \lambda$.  Using existing results on the security of the Fiat-Shamir transform, we have the following (informal) corollary.
	
	\begin{corollary}
		Stern's optimized signature scheme with salt + index is EUF-CMA secure
	\end{corollary}
	
	Let us discuss the technical difficulties we had to overcome in order to prove our theorem. It seems at first strange that the advantage of solving the syndrome decoding problem appears in the statement of our theorem but we will make this clear. We want to construct a simulator that will output transcripts computationally indistinguishable from real transcripts. In the identification scheme, the prover commits to some strings and reveals some of them. How should we simulate the committed values $comm = \hh(s,\salt,index)$? The natural solution is to simulate this with a random string. Indeed,  we model $\hh$ as a random function in the random oracle model so the commitment is a random string. This reasoning is not entirely correct. Since $\salt$ and index are publicly known ($\salt$ is revealed most often by another opening) it suffices for an adversary to guess $s$ in order to distinguish $c$ from a random string. This shows the difficulty of such zero-knowledge proofs with deterministic commitment, even without pseudo-random strings and seeds. As we mentioned earlier, if we use probabilistic commitments, this problem disappears and we can show that the commitment is statistically indistinguishable from a random string. 
	
	Going back to our case, we said that $comm = \hh(z,\salt,index)$ can be distinguished from a random string if we can guess $z$. This will actually be the only way to distinguish $comm$ from a random string, even in the t-HVZK setting, as long as each $\salt$ is different (which happens wp. $O(\frac{t^2}{2^{\saltsize}})$). Stern's identification scheme consists of many repetitions of an identification scheme where the prover commits to $3$ strings $z_1,z_2,z_3$ and reveals $2$ out of the $3$ depending on the verifier's challenge. In this context, guessing the string $z$ means guessing the string that has not been revealed. But we show that if the adversary can guess this string, then he can recover the secret key, meaning he can solve the underlying syndrome decoding problem. Regarding seeds, we show that if we have salts, the only way to distinguish the random strings from the pseudo-random strings is to guess a seed value, which happens wp. $2^{-\seedsize} = 2^{-\lambda}$ so with the $\salt$ + index construction, seeds of size $\lambda$ suffice. 
	\subsubsection{Adding other optimizations}
	
	Now that we have theoretic foundations that show the security of our optimized Stern's signature scheme. We propose some instantiations where we also add optimizations. 
	
	\paragraph{Changing the metric.} Stern's identification scheme was originally constructed based on the hardness of the binary syndrome decoding problem but actually, Shamir already proposed a similar scheme based on the permuted kernel problem, which can be seen as a generalization of the binary syndrome decoding. Several variants of this problem: ternary syndrome decoding with large weight \cite{DST19}, restricted syndrome decoding, syndrome decoding with Lee's weight \cite{CDE21,WKHBSP22,HTW21} have been proposed and studied. In \cite{CDE21}, the current authors and Thomas Debris-Alazard presented a cryptanalysis of syndrome decoding using Lee's metric with the idea of proposing Stern's signature scheme 
	Similar schemes but with different optimizations) have also been proposed. Here, we show that Stern's optimized scheme with Lee's metric gives better results than Stern's original scheme.
	
	\paragraph{Hash trees.} Once we allow ourselves seeds, there are actually many ways to combine seeds. Here, we show an optimization which is a good compromise between performance and space gain: we regroup $4$ rounds of the identification scheme together and we construct the seeds and commitments according to a Merkle tree. We show that this can also lead to some gains.

	\section{Hard Problems in Code-Based Cryptography}
	
	In this work, we focus on two hard problems in code-based cryptography: the Syndrome Decoding (SD) problem and the Permuted Kernel Problem(PKP). The two  have many properties in common, namely, both of them are known to be NP-complete~\cite{BMvT78,GJ90} and they are believed to be hard on average once the instance of the problem is taken from the appropriate distribution.

	\subsection{Syndrome Decoding (SD) problem}
	
	The canonical hard problem used in code-based cryptography is the syndrome decoding problem, defined as follows.
	
	\begin{problem}[Syndrome Decoding, $\SD(n,k,w)$]\label{Definition:SyndromeDecodignProblem}
		
		\textbf{Input:} A matrix $\Hm \in \F_q^{(n-k)\times n}, \ n, k \in \N$, a column vector (the syndrome) $\vect{s} \in \F_q^{n-k}$, a weight function $wt(\cdot): \F_q^n \rightarrow \N$, and a weight $w \in \N$. 
		
		\textbf{Goal:} Find a column vector $\ev \in \F_q^n$ that satisfies $\Hm\ev  = \sv$ and $wt(\ev)= w$.
	\end{problem}

	 We call $q$ the alphabet size of the problem. The original version of the problem, which we refer to as the binary syndrome decoding problem, corresponds to the case where $q = 2$ and the weight function is Hamming weight, defined as: $wt_H(\ev) = |\{i : e_i \neq 0\}|$. The decision version of this problem, asking whether there exists a vector $\ev$ of Hamming weight $w$ satisfying $\Hm \ev = \sv$, is proven to be $\NP$ complete\cite{BMvT78}. This is problem is also believed to be hard on average, as summarized by the following conjecture. 
	
	\begin{conjecture}[Average-case hardness of syndrome decoding]\label{Conjecture:AvgHardnessBSD}
		
		For suitable choices of $R = \frac{k}{n} = \Theta(1)$ and $\omega = \frac{w}{n} = \Theta(1)$, $SD(n,k,w)$ defined over the Hamming weight $wt(\cdot)$ is hard on average when the input $(\Hm,\sv)$ is sampled from  the distribution $\mathcal{D^1}_{n,k,w}$ given as:
		\begin{align}\label{Equation:AverageDistribution}
			\mathcal{D^1}_{n,k,w} :\Hm \xleftarrow[]{\$} \F_q^{(n-k)\times n}, \ev \xleftarrow[]{\$} S_w, \ \textrm{return } (\Hm, \sv = \Hm \ev).
		\end{align}
	where $S_w$ is the set of vectors of weight $w$.
	\end{conjecture}
	
	 Moreover, the hardness holds even in the presence of quantum computers. One has to be careful with the choice of parameters. Indeed,  for constant values $0 < R,\omega < 1$, this problem can be solved in time $poly(n)$ using Prange's algorithm~\cite{Pra62} for $\omega \in [(1-R)\frac{q-1}{q},1 - (1-R)\frac{q-1}{q}]$, while the problem is believed to be exponentially hard outside this zone even for quantum computers, see for instance \cite{KT17}.\\

	\subsection{Permuted Kernel Problem (PKP)}
	The Permuted Kernel Problem was introduced by Shamir~\cite{Sha89}, and has the same linear algebra flavor as the syndrome decoding problem. The main difference is that the weight constraint is replaced with a combinatorical constraint, which seemingly removes the connection with decoding linear codes, while preserving the NP-completeness of the problem. The problem is defined as follows.\\
	
	\begin{problem}[Permuted Kernel Problem $\PKP(n,k,\vv)$]\label{Problem:PermutedKernelProblem}
		
		\textbf{Input:} A matrix $\Hm \in \F_q^{(n-k)\times n}, \ n, k \in \N$, a column vector $\vec{s} \in \F_q^{n-k}$, a column vector $\vv \in \F_q^n$. 
		
		\textbf{Goal:} Find a permutation $\sigma$ acting on $[n]$ that satisfies $\Hm \sigma(\vv) = \sv$.
	\end{problem}
	
	 As in the case of the syndrome decoding problem, it is believed that the permuted kernel problem is hard on average for suitable choices of $\vv$, i.e. given the instances of the problem is sampled from the distribution $\mathcal{D}_{n,k}(\vv)$ sampled as follows:
	\begin{align*}
		\mathcal{D}_{n,k}(\vv) :\Hm \xleftarrow[]{\$} \F_q^{(n-k)\times n}, \sigma \xleftarrow[]{\$} \Perm_{[n]}, \textrm{ return } (\Hm,\sv = \sigma(\vv),\vv).
	\end{align*}
	
	 One can see that the permuted kernel problem is similar to the syndrome decoding one in many respects. In fact, if we observe the binary syndrome decoding problem over the Hamming weight, and the corresponding permuted kernel problem over the binary field, we can easily show that the two are equivalent.\footnote{If the weight constraint from the syndrome decoding problem, $wt_H(\ev) = w$, is replaced with the permutation constraint from the permuted kernel problem, $\ev = \sigma(\vv)$ (where $\vv$ is given and it satisfies $wt_H(\vv) = w$), a solution to first problem yields a solution to the second, and vice versa.} For other weight functions, however, they are not equivalent as the set of potential solutions of the syndrome decoding problem is generally wider than the set of potential solutions of the permuted kernel problem\footnote{The number of non-binary words of weight $w$, length $n$, and alphabet of size $q$ is given as the number of permutations of any word of a given weight and length multiplied by the number of compositions of $w$ into $n$ parts taking values in $\{0,\dots,q-1\}$. In the binary case, there is only one possible such composition, while in the non-binary case the number of such compositions can be, and in most cases it is, greater than $1$.}. Nevertheless, it can be shown that for the properly chosen parameters, finding a solution of the permuted kernel problem yields a solution to the syndrome decoding problem, as summarized in the following proposition.\\
	
	\subsection{Reduction from SD to PKP}
		 Other versions of syndrome decoding problem, namely those with a weight function other than the Hamming weight and an alphabet size greater than $2$, start to attract more attention of the cryptographic community in the recent years as they yield potentially harder computational problems and hence more efficient code-based schemes. One example of a weight function that is known from coding theory and that recently gained more attention in cryptography is the Lee weight $wt_L$, defined as: $\forall \textbf{e} = (e_0,...,e_{n-1}) \in \mathbb{F}_q^n, \quad wt_L(\textbf{e}) = \sum_{i = 0}^{n-1} \min(e_i, q - e_i).$ The hardness of $SD(n,k,w)$ using Lee's weight has been studied against classical adversaries as well as the quantum ones~\cite{CDE21,WKHBSP22} in the Information Set Decoding framework and it is shown that problem is computationally hard in both cases.

	\begin{definition}[Weight function from distance function]\label{Definition:WeightFunction} 
		We say that a weight function $wt : \F_q^n \rightarrow \mathbb{N}$ is constructed from a distance function is there exists  a distance function $d : \F_q \times F_q \rightarrow \mathbb{N}$ st. 
		\begin{align}
			 \forall \textbf{e} = (e_0,\dots,e_{n-1}) \in \mathbb{F}_q^n, \quad wt(\textbf{e}) = \sum_{i=0}^{n-1} d(e_i,0).
		\end{align}
	\end{definition}
	
	 This definition allows us to generalize the syndrome decoding problem over different metric spaces with potentially better security and/or better efficiency.

	\begin{definition}[Vector decomposition] Let $\xv \in \F_q^n$. 	The decomposition of $\xv$ is a $q$-tuple $\cv = (c_0, c_1, \dots, c_{q-1})$ in which $c_0$ counts the number of zeros in $\xv$, $c_1$ counts the number of ones, ..., and $c_{q-1}$ counts the number of $q-1$ elements in $\xv$.  In other words:
		$$c_i = |\{j \in \{0,\dots,n-1\} : x_j = i\}|.$$
	\end{definition}
	
	\begin{lemma}[\cite{Bo83}]\label{Lemma:VectorPermutations}
		The number of vectors in $\F_q^n$ having the same decomposition $\cv = (c_0, c_1, \dots, c_{q-1})$ is given by the multinomial coefficient  $\binom{n}{\cv} = \binom{n}{c_0, c_1, \dots, c_{q-1}} \eqdef \frac{n!}{c_0!c_1! \cdots c_{q-1}!}$. 
	\end{lemma}
	
	We now prove what is essentially a reduction from the SD problem to the PKP problem.
	
	\begin{proposition}[Reduction from SD to PKP]
		Take a constant alphabet size $q$. Assume we have an algorithm $\aa$ that can solve the $PKP(n,k,\vv)$ problem in time $t$. Then, we can use this algorithm to solve randomly chosen instance $(\Hm,\sv) \Unif D_{n,k,wt(\vv)}$ in time $t$  wp. $\ge \Omega(\frac{1}{\poly(q)})$. This is true for any weight function $wt()$ satisfying Definition~\ref{Definition:WeightFunction}.
	\end{proposition}
	
	\begin{proof}
		
		Let us fix an instance $(\Hm,\sv)$ of the syndrome decoding problem sampled from $\mathcal{D}_{n,k,w}$, and let $\ev \in S_w$ be the random vector st. $\Hm \ev = \sv$ sampled when constructing $(\Hm,\sv)$, as it is done in $\mathcal{D}_{n,k,w}$. Recall that $S_w \eqdef \{\xv \in \F_q^n: wt(\xv) = w\}$. Let also $C_w$ be the set of all decompositions of the elements from $S_w$, namely:
		\begin{align*}
			C_w \eqdef \{\cv = (c_0,\cdots,c_{q-1}) \in \mathbb{N}^q : \forall i \in \{0,\dots,q-1\},  \sum_{i = 0}^{q-1} c_i = n \textrm{ and } \sum_{i = 0}^{q-1} c_i d(0,i) = w\}.
		\end{align*}
		
		 The size of $S_w$ can now be calculated as the number of vectors having a decomposition in $C_w$, summing over $\cv \in C_w$. Using \ref{Lemma:VectorPermutations}, we  get
		\begin{align*}
			|S_w| = \sum_{\vec{c}\in C_w} \binom{n}{\vec{c}}.
		\end{align*}
		
		Now, let $C$ be the set of all decompositions, namely
		\begin{align*}
			C \eqdef \{\cv = (c_0,\cdots,c_{q-1}) \in \mathbb{N}^q : \forall i \in \{0,\dots,q-1\},  \sum_{i = 0}^{q-1} c_i = n\}.
		\end{align*}
		 The size of this set $C$ is known to be $\binom{n+q-1}{q-1}$\cite{Bo83}. Moreover, we have $C = \bigcup_w C_w$ hence $|C_w| \le |C|$ for each $w$. This implies
		\begin{align*}
			|S_w| = \sum_{\cv\in C_w} \binom{n}{\cv} \leq\binom{n+q-1}{q-1}\mathop{\max}\limits_{\cv\in C_w} \binom{n}{\cv}.
		\end{align*}
		
		 Let $\tilde{\cv}$ st. $\mathop{\max}\limits_{\cv \in C_w} \binom{n}{\cv} = \binom{n}{\tilde{\cv}}$, and let $\tilde{\vv}$ be a vector that has the decomposition given by $\tilde{\cv}$. We then fix an instance of the permuted kernel problem, $\PKP(n, k,\widetilde{v})$, for which the parity check matrix, and the syndrome are the same as for the syndrome decoding instance.
		
		 We run our algorithm $\aa$ to find a solution of this PKP instance we constructed. A solution to this problem will directly yield a solution to our original $SD$ problem, but it is possible that there is no solution to this PKP problem. Notice however that if the error vector $\ev$ (given by the $SD(n,k,w)$ instance) has the decomposition $\tilde{\cv}$, then our $PKP(n,k)$ instance will have a solution. Given that $\ev$ is sampled uniformly at random from the set $S_w$, this probability is then given by: $$\frac{\binom{n}{\tilde{\cv}}}{|S_w|} \geq \frac{\binom{n}{\tilde{\cv}}}{\binom{n+q-1}{q-1} \binom{n}{\tilde{\cv}}} = \frac{1}{\binom{n+q-1}{q-1}} \geq \frac{1}{(n+q-1)^{q-1}} = \Omega \Big(\frac{1}{\poly(n)}\Big),$$
	\end{proof}

\section{Identification and digital signature schemes}
\subsection{Identification schemes}
\subsubsection{First definitions}
In this article, we study $3$-round identification schemes, also known as $\Sigma$-protocols, defined as follows.
\begin{definition}[3-round identification scheme]
	A $3$-round identification scheme is an interactive protocol consisting of the following $3$ algorithms, which should be efficient with respect to some security parameter $\lambda$:
\begin{itemize} 
	\item a key generation algorithm $\ISKeygen(1^\lambda) \rightarrow (pk,sk)$,
	\item the prover's algorithms $\ISProver = (P_1,P_2)$ satisfying $P_1(sk) \rightarrow (x,St)$ and $P_2(sk,x,c,St) \rightarrow z$, where $x$ is the first message, $St$ is some internal state, $c$ is the challenge from the verifier and $z \in R$ the prover's response (i.e. second message),
	\item the verifier's algorithm $\ISCheck(pk,x,c,z)$ that outputs $1$ if the verifier accepts the interaction and $0$ otherwise.
\end{itemize}
\end{definition}

 We describe below the main steps of an identification scheme.
\begin{center}\cadre{\begin{center}
			Identification scheme $\IS$
		\end{center}
		
		\textbf{Initialization.} Let $(pk,sk) \leftarrow \ISKeygen(1^\lambda)$. The prover gets $(pk,sk)$ and the verifier gets $pk$.
		
		\textbf{Interaction.}
		\begin{enumerate}
			\setlength\itemsep{-0.2em}
			\item The prover generates $(x,St) \leftarrow P_1(sk)$ and sends $x$ to the verifier.
			\item The verifier picks a uniform $c$ from the challenge set $C$ and sends $c$ to the prover.
			\item The prover generates $z \leftarrow P_2(sk,x,c,St)$ and sends $z$ to the verifier.
		\end{enumerate}
		
		\textbf{Verification.} The verifier accepts iff. $\ISCheck(pk,x,c,z) = 1$. \\
}\end{center}
$ \ $ \\

 A transcript consists of the messages exchanged between the prover and the verifier when running $\IS$. In our setting, a transcript is described by $(x,c,z)$. We now define different distribution sof transcripts.

\begin{definition}[Distributions of transcripts, honest behavior]
	For an identification scheme $\IS$, $trans_{\IS}(pk,sk)$, is the distribution of transcripts when both the prover and the verifier are honest. This distribution is sampled as follows:
	$$trans_{IS}(pk,sk): (x,St) \leftarrow P_1(sk), c \Unif C, z \leftarrow P_2(sk,x,c,St), \mathrm{ return } (x,c,z).$$
\end{definition}
We also consider transcripts in the case an adversary $\aa$ who only knows $pk$ tries to impersonate the prover.
\begin{definition}
	For an identification scheme $\IS$ and an adversary $\aa(pk)$, $trans^{\aa}_{\IS}(pk)$, is the distribution of transcripts when an adversary $\aa =(\aa_1,\aa_2)$ that only knows the public key interacts with the verifier.
	$$trans^{\aa}_{\IS}(pk): (\tilde{x},\tilde{St}) \leftarrow \aa_1(pk), c \Unif C, \tilde{z} \leftarrow \aa_2(pk,\tilde{x},c,\tilde{St}), \mathrm{ return } (\tilde{x},c,\tilde{z}).$$
\end{definition}

\subsubsection{Completeness,soundness and zero-knowledge}
 An identification must satisfy the completeness, soundness and zero-knowledge properties that we describe below.

\begin{definition}[Completeness]
	An identification scheme $\IS$ is perfectly complete\footnote{In this paper, we construct a scheme that is perfectly complete. In general, this condition can be relaxed so that "almost perfectly complete scheme", where the probability above is very close to $1$, can also be considered as a proper identification scheme.} iff.
	\begin{align*}
		\Pr\left[\ISCheck(pk,x,c,z) = 1 \left| \substack{(pk,sk) \leftarrow \ISKeygen(1^\lambda) \\ (x,c,z) \leftarrow trans_{\IS}(pk,sk)} \right.\right] = 1.
	\end{align*}
\end{definition}

For the soundness property, we define the soundness advantage as follows

\begin{definition}[Soundness Advantage]
	Let $\IS$ be an identification scheme. For any adversary $\aa$
	\begin{align*}
		Adv_\IS(\aa) = \Pr\left[\ISCheck(pk,{x},c,{z}) = 1 \left| \substack{(pk,sk) \leftarrow \ISKeygen(1^\lambda) \\ ({x},c,{z}) \leftarrow trans_A(pk)} \right.\right].
	\end{align*}
\end{definition}

Notice that we write these definitions in order to perform concrete security statements and not asymptotic statements, where we would only require that an adversary that runs in time $poly(\lambda)$ has soundness advantage $negl(\lambda)$. \\

The zero-knowledge property ensures that the transcripts generated by the identification scheme do not reveal any information, in particular about the secret key. The way the notion of not revealing any information is formalized is through the existence of an efficient simulator $Sim(pk)$ who can output transcripts indistinguishable from real transcripts. For a specified simulator, $Adv^{HVZK}_\IS(\aa)$ corresponds to the probability that an adversary $\aa$ distinguishes a real transcript from a simulated transcript. 

\begin{definition}[Honest verifier Zero-knowledge advantage]
	Let $\IS$ be an identification scheme for which we specified a simulator $Sim$. For any adversary $\aa$, we define 
	$$ Adv^{HVZK}_\IS(\aa) = 
	\left|\Pr[b = 1 \left| \substack{(pk,sk) \leftarrow \ISKeygen(1^\lambda) \\ \tau = (x,c,z)\leftarrow trans_{\IS}(pk,sk) \\ b \leftarrow \aa(\tau,pk)}\right. ] - 
	\Pr[b = 1 \left| \substack{(pk,sk) \leftarrow \ISKeygen(1^\lambda)  \\ \tau = (x,c,z) \leftarrow Sim(pk) \\ b \leftarrow \aa(\tau,pk)}\right. ] \right|.$$
\end{definition}

As we said in the introduction, an important concept will actually be the notion of multi-HVZK, where we give to an adversary $t$ real or simulated transcripts and we want to consider the distinguishing advantage in this case. 

\begin{definition}[Multi-Honest verifier Zero-knowledge advantage]
	Let $\IS$ be an identification scheme for which we specified a simulator $Sim$. For any adversary $\aa$, we define 
	$$ Adv^{t-HVZK}_\IS(\aa) = 
	\left|\Pr[b = 1 \left| \substack{(pk,sk) \leftarrow \ISKeygen(1^\lambda) \\ \tau_1,\dots,\tau_t \leftarrow trans_{\IS}(pk,sk) \\ b = \aa(\tau_1,\dots,\tau_t,pk)}\right. ] - 
	\Pr[b = 1 \left| \substack{(pk,sk) \leftarrow \ISKeygen(1^\lambda)  \\ \tau_1,\dots,\tau_t \leftarrow Sim(pk) \\ b = \aa(\tau_1,\dots,\tau_t,pk)}\right. ] \right|.$$
\end{definition}

\subsubsection{Parallel repetition of an identification scheme}
For any identification scheme $\IS$, we define $\IS^{\otimes r}$ as the $r$-fold parallel repetition of $\IS$, which consists of the following:

\begin{center} \cadre{\begin{center}
			Identification scheme $\IS^{\otimes r}$
		\end{center}
		
		\textbf{Initialization.} Let $(pk,sk) \leftarrow \ISKeygen(1^\lambda)$. The prover gets $(pk,sk)$ and the verifier gets $pk$.
		
		\textbf{Interaction.}
		\begin{enumerate}
			\setlength\itemsep{-0.2em}
			\item For each $i \in [r]$, the prover generates $(x^i,St^i) \leftarrow P_1(sk)$ and then sends $\xv = x^1,\dots,x^r$ to the verifier.
			\item The verifier picks a random $\mathbf{c} = c^1,\dots,c^r$ for independently randomly chosen challenges and sends $\mathbf{c}$ to the prover.
			\item The prover generates $\zv = (z^1,\dots,z^r)$, where for each $i \in [r]$, $z^i \leftarrow P_2(sk,x^i,c^i,St^i)$, and then sends $\zv$ to the verifier.
		\end{enumerate}
		
		\textbf{Verification.} The verifier accepts iff. $\forall i \in [r], \ISCheck(pk,x^i,c^i,z^i) = 1$. 
} \end{center}

\subsection{Signature schemes}\label{Section:Fiat-Shamir}
A signature scheme $S$ consists of $3$ algorithms  $(\SKeygen,\SSign,\SVerify)$:
\begin{itemize}
	\item $\SKeygen(1^\lambda) \rightarrow (pk,sk)$ is the generation of the public key $pk$ and the secret key $sk$ from the security parameter $\lambda$.
	\item $\SSign(m,pk,sk) \rightarrow \sigma_m$ : generates the signature $\sigma_m$ of a message $m$ from $m,pk,sk$.
	\item $\SVerify(m,\sigma,pk) \rightarrow \zo$ verifies that $\sigma$ is a valid signature of $m$ using $m,\sigma,pk$. The output $1$ corresponds to a valid signature.
\end{itemize}
We require from our signature schemes correctness and the EUF-CMA property.
\begin{definition}
A signature scheme is correct iff. when we sample $(pk,sk) \leftarrow \SKeygen(1^\lambda)$, we have for each $m$
$$ \Pr[\SVerify(m,\SSign(m,pk,sk),pk)] = 1.$$
\end{definition}
We consider the standard EUF-CMA security for signature schemes. To define the advantage of an adversary $\aa$, we consider the following EUF-CMA game \\ \\

\cadre{
	EUF-CMA game \\
	\hline
	$(pk,sk) \leftarrow \SKeygen(1^\lambda)$. \\
	Run
	$(m^*,\sigma^*) \leftarrow \aa(pk)$, where 
	$\aa(pk)$ can perform $q$ sign queries. More precisely, $\aa(pk)$ constructs $m_1,\dots,m_q$ and receives $\sigma_1 \leftarrow \SSign(m_1,sk,pk),\dots,\sigma_q \leftarrow\SSign(m_q,sk,pk)$. Then $\aa$ performs a procedure to output $(m^*,\sigma^*)$. \\
	$\aa$ wins the game iff. $\SVerify(m^*,\sigma^*,pk) = 1 \wedge \forall i \in [q], m^* \neq m_i$.
	} $ \ $ \\

 The EUF-CMA advantage is defined as follows
\begin{definition}
	For a signature scheme $S$ and an adversary $\aa$, we define 
	$$Adv_S^{\textrm{EUF-CMA}}(\aa) = \Pr[\aa \mbox{ wins the } \textrm{EUF-CMA } \mathrm{ game}].$$
\end{definition}

\subsubsection{The Fiat-Shamir transform} We can directly construct a signature scheme from an identification scheme via the Fiat-Shamir transform~\cite{FS87}. From an identification scheme $\IS = (\ISKeygen,\ISProver = (P_1,P_2),\ISCheck)$, we define the following signature scheme $\ss = (\SKeygen,\SSign,\SVerify)$ that uses a random function $\hh$:  \\ \\
\cadre{The Fiat-Shamir transform: constructing $S$ from $\IS$ \\ \hline 
	\begin{itemize}
		\item $\SKeygen(1^\lambda) = \Init(1^\lambda)$
		\item $\SSign(m,pk,sk) : (x,St) \leftarrow P_1(pk), c \leftarrow \hh(x,m), z \leftarrow P_2(sk,x,c,St)$, output $\sigma = (x,z)$.
		\item $\SVerify(m,\sigma = (x,z),pk) = V(pk,x,\hh(x,m),z)$.
	\end{itemize}
} $ \ $ \\

The transformation is proven to be secure in both classical and quantum random-oracle model, if the identification has good properties. Recently in~\cite{GHHM21}, the Fiat-Shamir is proven to be secure essentially if the underlying identification scheme has low soundness advantage and low $t$-HVZK advantage. For Stern's identification scheme, the soundness advantage will be immediate and the main focus of our work will be on the $t$-HVZK advantage. For instance the attack we present later in the paper uses the fact that one of the proposed variants of Stern's identification has a small HVZK-advantage but a big $t$-HVZK advantage for $t$ large enough.

\subsection{Commitment schemes}
A commitment scheme consists of $3$ algorithms:
\begin{itemize}
	\item A commit function $Commit(s) \rightarrow (c,St)$. $c$ is the commitment and $St$ is some internal state (usually some private randomness), used for the opening.
	\item An open function $Open(s,St) \rightarrow o$ that outputs the opening of the string $s$. 
	\item A function that verifies the validity of an opening $o$ for a string $s$ and a commitment $c$: $Ver(c,s,o) \rightarrow \zo$.
\end{itemize}
For a commitment $c$, $(s,o)$ is valid iff. $Ver(c,s,o) = 1$. Ideally, we want $2$ properties from a commitment scheme. First, we want a commitment scheme to be hiding, meaning that any adversary $\aa$ receiving only the commitment $c$ shouldn't have information about the committed value. More precisely, $\forall s,s'$, we want the following advantage to be as small as possible:
\begin{align*}
	Adv^{Hiding}(\aa) = \left| 
	\Pr[b = 1 \left| \substack{(c,St) \leftarrow Commit(s) \\ b \leftarrow \aa(s,c)}\right.] - 
	\Pr[b = 1 \left| \substack{(c,St) \leftarrow Commit(s') \\ b \leftarrow \aa(s',c)}\right.] \right|
\end{align*}
Again, we deal with concrete security so we don't specify what condition we want for the hiding advantage. If we dealt with asymptotic security, we would require that for any adversary running in time $poly(\lambda)$ for some security parameter $\lambda$, the above advantage is negligible in $\lambda$.
The second property we want from a commitment scheme is the binding property, which means that an adversary cannot produce a commitment which can be revealed to $2$ different values. We therefore want the following advantage to be as small as possible:
\begin{align*}
	Adv^{Binding}(\aa) = \Pr[Ver(c,s,o) = Ver(c,s',o') = 1 \wedge s \neq s' \left| (c,s,o,s',o') \leftarrow \aa() \right.]
\end{align*}

\section{Stern's identification scheme}
\subsection{Description of the scheme}

In it's original form, the protocol is based on the binary syndrome decoding problem over the Hamming weight, where the instance of the syndrome decoding problem is taken from the distribution $\mathcal{D}_{n,k,w}$ defined in the previous section. As mentioned in the introductory part, different variants of the scheme were proposed to mitigate the problem of protocol's high communication cost. The most prominent ones, however, are based on the use of pseudo-random generators, so we focus on them in the next two sections.\\

 We describe the scheme for the case of binary syndrome decoding but our descriptions can be easily extended to the case of PKP with larger alphabet size $q$. We consider a commitment scheme $(Commit,Open,Ver)$ and let $Perm_n$ be the set of permutations acting on $[n]$. Stern's $1$ round identification protocol is the following:

\begin{openbox}[\textit{$1$ round Stern's identification scheme, generic commitment, $\IS_{Stern}^1$}]
	 \textbf{Initialization.} Let $(\Hm,\ev,\sv = \Hm\ev) \xleftarrow[]{\$} \mathcal{D}_{n,k,w}$ \ref{Conjecture:AvgHardnessBSD}, $sk = \ev$, $pk = (\Hm,\sv)$. The prover obtains $(pk,sk)$ and the verifier obtains $pk$.
	
	 \textbf{Interaction.}  
	\begin{itemize}
		\item The prover picks a random permutation $\pi \in Perm_n$ and a random column vector $\yv \xleftarrow[]{\$} \F_q^n$, and then computes $\tv = \Hm  \yv$. The prover computes $z_1 = (\pi,\tv)$, $z_2 = (\pi(\yv))$, $z_3 = (\pi(\yv + \ev))$. For $i \in \{1,2,3\}$, he then constructs $(x_i,St_i) \leftarrow Commit(z_i)$ and sends each $x_1,x_2$ and $x_3$ to the verifier.\\
		
		\item The verifier sends a random challenge $c \xleftarrow[]{\$} \{1,2,3\}$ to the prover.\\
		
		\item The prover reveals information so that the verifier can recover $z_i$ for each $i \in \{1,2,3\} \setminus c$. More precisely:
		\begin{enumerate}
			\item if $c = 1$, the prover computes $o_2 = Open(z_2,St_2), o_3 = Open(z_3,St_3)$ and sends $z_2,z_3,o_2,o_3$ to the verifier.
			\item if $c = 2$, the prover computes $o_1 = Open(z_1,St_1), o_3 = Open(z_3,St_3)$ and sends $z_1,z_3,o_1,o_3$ to the verifier.
			\item if $c = 3$, the prover computes $o_1 = Open(z_1,St_1), o_2 = Open(z_2,St_2)$ and sends $z_1,z_2,o_1,o_2$ to the verifier.
		\end{enumerate}
	\end{itemize}	
	 \textbf{Verification.} 
	\begin{enumerate}
		\item if $c =1$, the verifier checks that $|z_2 + z_3| = w$ and checks the validity of the commitments $z_2,z_3$ so he checks that $Ver(x_2,z_2,o_2) = Ver(x_3,z_3,o_3) = 1$.
		\item if $c = 2$, the verifier gets $z_1 = (\pi,\tv)$ and $z_3 = \pi(\yv + \ev)$. He checks that $\Hm(\pi^{-1}(z_3)) = \Hm(\yv + \ev) = \tv + \sv$. He also checks the validity of the commitments $z_1$ and $z_3$.
		\item if $c = 3$, the verifier gets $z_1 = (\pi,\tv)$ and $z_2 = \pi(\yv)$. He checks that $\Hm(\pi^{-1}(z_2)) = \Hm(\yv) = \tv$. He also checks the validity of the commitments $z_1$ and $z_2$.
	\end{enumerate}
\end{openbox}

\begin{definition}
	Stern's identification scheme $\IS_{Stern}$ is the $R$-fold parallel repetition of $(\IS_{Stern}^1)$ described above, namely $(\IS_{Stern}^1)^{\otimes R}$. When aiming for $\lambda$ bits of classical security, we take $R = \lambda \log_2(3)$, and we double the number of rounds $R$ if we want $\lambda$ bits of quantum security.
\end{definition}

\subsection{Basic properties of Stern's identification scheme}
Stern's identification scheme is clearly complete. Soundness and HVZK properties are also known from this scheme and a thorough analysis has been done for instance in~\cite{Lei18}. Regarding soundness, we have the following
\begin{proposition}
	$$ Adv_{\IS_{Stern}^{1}}(\aa) \le \frac{2}{3} + Adv^{Binding}(\aa) + Adv^{SD}(\aa).$$
	where $Adv^{SD}$ denotes the probability to break the underlying syndrome decoding problem. 
\end{proposition}
\begin{proof}
	The way this statement is proved actually by using a notion called $3$-special soundness. We show that if an adversary can win the identification scheme simultaneously for the $3$ challenges then: or he broke the binding property of the commitment function, or solved the underlying syndrome decoding problem. We refer the reader to~\cite{Lei18} for details of this proof.
\end{proof}

Moreover, the above scheme is zero-knowledge, if the commitment scheme is hiding
\begin{proposition}~\label{Proposition:1roundHVZK}
	$$ Adv_{\IS_{Stern}^{1}}^{HVZK}(\aa) \le Adv^{Hiding}(\aa).$$
\end{proposition}
\begin{proof}
	We consider the following simulator:
	\begin{openbox}[\textit{Simulator}]
		\noindent {\textbf{Input.}} A public key ${pk} = (\Hm,\sv)$. \\
		\textbf{Execution.} Pick a random challenge $c \Unif \{1,2,3\}$.
		\begin{itemize}
			\item If $c = 1$, choose $z_2 \Unif \F_q^n,  y \Unif S_w$ and let $z_3 = z_2 + y$. For $i \in \{2,3\}, $ compute $(x_i,St_i) \leftarrow Commit(z_i)$ and $o_i = Open(z_i,St_i)$.
			\item If $c = 2$, choose $\pi \Unif Perm_n, \yv \Unif \F_q^n$ and set $\tv = \Hm \yv$, $z_1 = (\pi,\tv)$ and $z_2 = \pi(\yv)$. For $i \in \{1,3\}, $ compute $(x_i,St_i) \leftarrow Commit(z_i)$ and $o_i = Open(z_i,St_i)$.
			\item If $c = 3$, choose $\pi \Unif Perm_n, \yv' \Unif \F_q^n$ and set $\tv = \Hm \yv' - \sv$, $z_1 = (\pi,\tv)$ and $z_3 = \pi(\yv')$. For $i \in \{1,2\}, $ compute $(x_i,St_i) \leftarrow Commit(z_i)$ and $o_i = Open(z_i,St_i)$.
		\end{itemize}
		\ \ \quad Compute $(x_c,St_c) \leftarrow Commit(\mathbf{0})$ where $\mathbf{0}$ is the all $0$ string.\\
		\textbf{Output.}  $\tau = (x_1,x_2,x_3,c,o_{c'},o_{c''})$ where $c' \neq c''$ and $c',c'' \neq c$.
	\end{openbox}
	
	We also rewrite the procedure that generates real transcripts 
	
	\begin{openbox}[\textit{Real transcript}]
		\noindent {\textbf{Input.}} A public key ${pk} = (\Hm,\sv)$. A secret key $sk = \ev$. \\
		\textbf{Execution.} Pick a random permutation $\pi \in Perm_n$ and a random string $y \in \F_q^n$. Let $\tv = \Hm \yv$ and $z_1 = (\pi,\tv), z_2 = \pi(\yv), z_3 = \pi(\yv + \ev)$. Pick a random challenge $c \Unif \{1,2,3\}$. For $i \in \{1,2,3\}$, let $(x_i,St_i) \leftarrow Commit(z_i)$ and $o_i = Open(z_i,St_i)$. \\
		\textbf{Output.}  $\tau = (x_1,x_2,x_3,c,o_{c'},o_{c''})$ where $c' \neq c''$ and $c',c'' \neq c$.
	\end{openbox}

We now prove that distinguishing these $2$ distributions is equivalent to breaking the hiding property of the commitment scheme. First, notice that $c \Unif \{1,2,3\}$ in both cases. Now, for each challenge, we look at the distribution of $z_{c'},z_{c''}$ for the $2$ values $c',c'' \in \{1,2,3\}$ which are different from $c$.
\begin{itemize}
	\item If $c = 1$: both for the real and the simulated transcript, $z_2$ is a random string in $\F_q^n$ and $z_3 - z_2$ is a random string in $S_w$
	\item If $c = 2$: both for the real and the simulated transcript, $z_3$ is a random string in $\F_q^n$. Regarding $z_1 = (\pi,\tv)$, in both cases $\pi$ is a random permutation in $Perm_n$ and $\tv = \pi^{-1}(z_2) - \sv$.
	\item If $c = 3$, both for the real and the simulated transcript, $z_2$ is a random string in $\F_q^n$. Regarding $z_1 = (\pi,\tv)$, in both cases $\pi$ is a random permutation in $Perm_n$ and $\tv = \pi^{-1}(z_2)$.
\end{itemize}
The same distribution for $z_{c'},z_{c''}$ imply the same distributions for $(x_{c'},x_{c''},o_{c'},o_{c''})$ given $c$ for real and simulated transcripts. Finally, we are left with $x_c$. For real transcripts, it is the commitment of $z_c$ and for simulated transcripts, it is the commitment of $\mathbf{0}$. By definition of the hiding property, if we can distinguish these $2$ settings then we can break the hiding property of the commitment scheme.  
\end{proof}

\subsection{Optimized schemes}
\subsubsection{The choice of commitment}

From a hash function $\hh$, it is known how to construct a commitment scheme which is both hiding and binding. We take $Commit(s) = (c = \hh(s||\rho), St = \rho)$ for a randomly chosen $\rho$ of large enough length, and the opening consists of revealing the value $\rho$, from which one can verify that the commitment is well formed. However, this implies that in order to reveal $s$, one has to send this string $\rho$ each time. We will call this commitment a randomized commitment

In the literature regarding Stern's signature scheme, we most often see another choice for commitment, namely $Commit(s) = (c = \hh(s),St = \emptyset)$ and the opening is also empty so in order to open $s$, you can just reveal $s$ and the verification just checks that $c = \hh(s)$. This improves the communication cost. However, this scheme is not hiding anymore. Indeed, an adversary can distinguish $Commit(s)$ from $Commit(s')$ for fixed (and known) $s,s'$ just by recomputing these $2$ values (which he couldn't do before because he didn't know $\rho$). We will call this commitment a deterministic commitment. 

We haven't seen the choice of deterministic commitment widely discussed regarding security. It seems that this has been the default choice (except for example in~\cite{Lei18} and in MQDSS which is not based on the syndrome decoding problem but for which a similar discussion exists) because of it's higher efficiency and because the lack of hiding property didn't seem to be exploitable. It could be that the hiding property given by randomized commitments is overkill and that deterministic commitments are enough for having a secure scheme. This question is actually at the core of our contributions, where we show that there are security issues regarding deterministic schemes, but that they can be repaired much more efficiently than by using the randomized commitments described above.

\subsubsection{Adding seeds}

 Another standard technique for reducing the communication cost of identification protocols is to replace some of the random vectors with pseudo-random ones. We thus consider the following functions, i.e. pseudo-random generators: $$E_0 : \zo^{\seedsize} \rightarrow \zo^{\seedsize} \times \zo^{\seedsize} ,\  E_1 : \zo^{\seedsize} \rightarrow Perm_n, \ E_2 : \zo^{\seedsize}  \rightarrow \F_q^n,$$ where $\seedsize \in \N$ is the seed size.
We present below a state of the art $1$ round optimized scheme which uses deterministic commitments and seeds, which is similar to the one presented in~\cite{BGMS22}. We denote by $\hh_{comm}$ the hash function used for the commitment.

\begin{openbox}[\textit{Stern's identification scheme (optimized, $1$ round), $\IS_{Stern,opt}^1$}]
	 \textbf{Initialization.} Let $(\Hm,\ev,\sv = \Hm\ev) \xleftarrow[]{\$} \mathcal{D}_{n,k,w}$ \ref{Conjecture:AvgHardnessBSD}, $sk = \ev$, $pk = (\Hm,\sv)$. The prover obtains $(pk,sk)$ and the verifier obtains $pk$.\bigskip
	
	 \textbf{Interaction.}  
	\begin{itemize}
		\item The prover picks a random $\seed \in \zo^{\seedsize}$. Let $(\seed_{\pi},\seed_{\yv}) = E_0(\seed), \pi = E_1(\seed_\pi)$ and $\yv = (\pi)^{-1}(E_2(\seed_{\yv}))$. It then computes $\tv = \Hm  \yv$, $z_1 = (\pi,\tv),z_2 = \pi(\yv), z_3 = \pi(\yv + \ev)$, and computes for $i \in \{1,2,3\}, \ x_i = \hh_{comm}(z_i)$. He then sends $x_1,x_2$ and $x_3$ to the verifier.\\
		
		\item The verifier sends a random challenge $c \xleftarrow[]{\$} \{1,2,3\}$ to the prover.\\
		
		\item The prover reveals information so that the verifier can recover $x_i$ for each $i \in \{1,2,3\} \setminus c$. More precisely:
		\begin{enumerate}
			\item if $c = 1$, the prover sends $l_1 = (o_2,o_3)$ with $o_2 = \seed_{\yv}, \ o_3 = \pi(\ev)$.
			\item if $c = 2$, the prover sends $l_2 = (o_1,o_3)$ with $o_1 = \seed_{\pi}, \ o_3 = \yv + \ev$.
			\item if $c = 3$, the prover sends $l_3 = \seed$.
		\end{enumerate}
	\end{itemize}	
	 \textbf{Verification.} 
	\begin{enumerate}
		\item if $c =1$, the verifier recovers $z_2,z_3$ from $o_2,o_3$ by taking $z_2 = E_2(o_2)$ and $z_3 = z_2 + o_3$. Then for $i \in \{2,3\}$ checks that $\hh_{comm}(x_i) = z_i$ and that $wt(o_3) = w$.
		\item if $c = 2$, the verifier computes $ \pi = E_1(o_1), z_3 = \pi(o_3)$ then $\tv = \Hm o_3 - \sv$ and finally $z_1 = (\pi,\tv)$. Then for $i \in \{1,3\}$ checks that $\hh_{comm}(x_i) = z_i$.
		\item if $c = 3$, the verifier computes $(\seed_{\pi},\seed_y) = E_0(l_3)$,  $\pi = E_1(\seed_\pi)$ and $z_2 =E_2(\seed_{\yv})$, followed by $\tv = \Hm (\pi)^{-1}(z_2)$. From there, he constructs $z_1 = (\pi,\tv)$.  Then for $i \in \{1,2\}$ checks that $\hh_{comm}(x_i) = z_i$.
	\end{enumerate}
\end{openbox}

Notice that with deterministic commitments, there are here other optimizations regarding what the prover sends for each challenge. He doesn't reveal the full strings $z_{c'},z_{c''}$ but the minimal information required to recover them. 

\subsubsection{Replacing the first message with a global commitment}
Recall that the full scheme consists of $R$ parallel repetition of the above $1$ round schemes. In this case, it is possible to replace the first message by a global commitment of all the $x_i$. This allows to reveal the commitment only of the values not revealed (since they can be recovered from the openings). The scheme becomes the following.

\begin{openbox}[\textit{Stern's identification scheme (optimized, $R$ rounds), $\IS_{Stern,opt}$}]
	\textbf{Initialization.} Let $(\Hm,\ev,\sv = \Hm\ev) \xleftarrow[]{\$} \mathcal{D}_{n,k,w}$ \ref{Conjecture:AvgHardnessBSD}, $sk = \ev$, $pk = (\Hm,\sv)$. The prover obtains $(pk,sk)$ and the verifier obtains $pk$.\bigskip
	
	\textbf{Interaction.} 
	\begin{itemize}
		\item  For $j$ from $1$ to $R$: the prover picks a random $\seed^j \in \zo^{\seedsize}$ and lets $(\seed^j_{\pi},\seed^j_{\yv}) = E_0(\seed^j), \pi^j = E_1(\seed^j_\pi)$ and $\yv^j = (\pi^j)^{-1}(E_2(\seed^j_{\yv}))$. It then computes $\tv^j = \Hm  \yv^j$, $z^j_1 = (\pi^j,\tv^i),z_2 = \pi^i(\yv^i), z^i_3 = \pi^i(\yv^i + \ev)$, and computes for $i \in \{1,2,3\}, \ x^j_i = \hh_{comm}(z^j_i)$.
	\end{itemize}
	The prover sends $\widetilde{x} = \hh_{\comm}(x_1^1,x_2^1,x_3^1,\dots,x_1^R,x_2^R,x_3^R)$ to the verifier.
	\begin{itemize}
		\item For $j$ from $1$ to $R$: the verifier sends a random challenge $c^j \xleftarrow[]{\$} \{1,2,3\}$ to the prover.\\
		
		\item For $j$ from $1$ to $R$: the prover reveals information so that the verifier can recover $x^j_i$ for each $i \in \{1,2,3\} \setminus c$. More precisely:
		\begin{enumerate}
			\item if $c^j = 1$, the prover sends $l^j_1 = (o^j_2,o^j_3)$ with $o^j_2 = \seed_{\yv}, \ o^j_3 = \pi^j(\ev)$ and $x^j_1$.
			\item if $c^j = 2$, the prover sends $l^j_2 = (o^j_1,o^j_3)$ with $o^j_1 = \seed_{\pi}, \ o^j_3 = \yv^j + \ev$ and $x^j_2$
			\item if $c^j = 3$, the prover sends $l^j_3 = \seed$ and $x_3^j$.
		\end{enumerate}
	\end{itemize}	
	\textbf{Verification.} For $j$ from $1$ to $R$: 
	\begin{enumerate}
		\item if $c^j =1$, recovers $z^j_2,z^j_3$ from $o^j_2,o^j_3$ by taking $z^j_2 = E_2(o^j_2)$ and $z^j_3 = z^j_2 + o^j_3$. Then for $i \in \{2,3\}$ checks that $\hh_{comm}(x^j_i) = z^j_i$ and that $wt(o^j_3) = w$.
		\item if $c^j = 2$, computes $ \pi^j = E_1(o^j_1), z^j_3 = \pi^j(o^j_3)$ then $\tv^j = \Hm o^j_3 - \sv$ and finally $z^j_1 = (\pi^j,\tv^j)$. Then for $i \in \{1,3\}$ checks that $\hh_{comm}(x^j_i) = z^j_i$.
		\item if $c^j = 3$, computes $(\seed^j_{\pi},\seed^j_y) = E_0(l^j_{3})$,  $\pi^j = E_1(\seed^j_\pi)$ and $z^j_2 =E_2(\seed^j_{\yv})$, followed by $\tv^j = \Hm (\pi^j)^{-1}(\yv^j)$. From there, he constructs $z^j_1 = (\pi^j,\tv^j)$.  Then for $i \in \{1,2\}$ checks that $\hh_{comm}(x^j_i) = z^j_i$.
	\end{enumerate}
	The verifier also checks that $\hh_{\comm}(x_1^1,x_2^1,x_3^1,\dots,x_1^R,x_2^R,x_3^R)$ is equal to the value $\widetilde{x}$ he received. 
\end{openbox}

Here, we described an identification scheme. The resulting signature schemes is constructed using the Fiat-Shamir transform presented in Section~\ref{Section:Fiat-Shamir}. This means each signature is of the form:

$$ \sigma(m) = (\widetilde{x},x^1_{c_1},\dots,x^R_{c_R},l^1_{c_1},\dots,l^R_{c_R}),$$
where $c = (c_1,\dots,c_R) = \hh_{\comm}(\widetilde{x}||M)$.

\section{Attack on optimized Stern's signature schemes}~\label{Section:Attack}
Here, we show that the combination of deterministic commitment and of seeds with $\seedsize = \lambda$ gives only $\lambda/2$ bits of security instead of the claimed $\lambda$ bits of security from~\cite{BGMS22}.

\begin{openbox}[Attack on the optimized signature scheme] 
	$$ \ $ 
\begin{itemize}
	\item The adversary performs $q_s = \lceil\frac{1}{R}2^{\seedsize/2}\rceil$ signature queries, with not necessarily distinct messages, to obtain the signatures $\sigma_{m_1},\dots,\sigma_{m_{q_s}}$ from the signer, where each $\sigma(m_k)$ is written as
	$$\sigma_{m_k} = \big(\widetilde{x}_k, x^1_{c_1,k},\dots,x^R_{c_R,k},l^1_{c_1,k},\dots,l^R_{c_R,k}) \big).$$
	\item The adversary tries to find a pair $((i,k),(i',k'))$ st. 
	\begin{align}\label{Eq:AttackCondition}(i,k) \neq (i',k') \wedge x_{2,k}^i = x_{2,k'}^{i'} \wedge c_k^i = 2 \wedge c_{k'}^{i'} \neq 2.\end{align}
	If none are found, return to the first step.
	
	\item From such a pair of couples $((i,k),(i',k'))$, the adversary recovers the secret key as follows: 
	\begin{itemize}
		\item Since $c^{i}_k = 2$, we have that from $l^i_2$, the adversary can recover $\pi^i_k = E_1(z^i_{1,k})$ and $\pi^i_k(\yv^i_k + \ev_k) = \pi(z^i_{3,k}).$
		\item Since $c^{i'}_{k'} \neq 2$, then the adversary obtains either $\seed^{i'}_{\yv, k'}$ or $\seed^{i'}_{k'}$ from which $\seed^{i'}_{\yv, k'}$ can be calculated. Moreover, as $x_{2, k'}^{i'} = x_{2, k}^i$, with overwhelming probability $\seed^i_{\yv, k} = \seed^{i'}_{\yv, k'}$, the adversary learns $\pi^i_k(\yv^i_k) = E_2(\seed^{i'}_{\yv, k'}).$\\
	\end{itemize}
	
	\item From $\pi^i_k, \ \pi^i_k(\yv^i_k), \ \pi^i_k(\yv^i_k + \ev)$, the adversary thus recovers the secret key by computing
	$$\ev = (\pi^{i}_k)^{-1}\left(\pi^i_k(\yv^i_k + \ev) - \pi^i_k(\yv^i_k)\right) = (\pi^{i}_k)^{-1}\left(\pi^i_k(\ev)\right),$$
	where $(\pi^{i}_k)^{-1}$ is the inverse of  $(\pi^{i}_k)$.\\
\end{itemize}
	
\end{openbox}

\begin{theorem}
	The attack presented above finds the secret key, $\ev$, in time $O(2^{\seedsize/2})$ doing $O(\frac{1}{r}2^{\seedsize/2})$ signature queries. The scheme thus preserves at most $\frac{\seedsize}{2}$ bits of security.
\end{theorem}

\begin{proof}
	
	 We  prove that each iteration of the attack described above finds the secret key with a constant probability. Notice that for $q_S = \lceil\frac{1}{R}2^{\seedsize/2}\rceil$, the signing oracle generates $Rq_S  = R\lceil\frac{1}{R}2^{\seedsize/2}\rceil \geq 2^{\seedsize/2}$ random seeds $\seed_{\yv,k}^i \in \{0,1\}^{l_{seed}}$, $i \in [R]$, $k \in [q_S]$. There is a collision in the seeds, i.e. $\exists (i, k), \ (i', k') \text{ such that }(i,k) \neq (i', k') \ \wedge \ \seed_{\yv,k}^i = \seed_{\yv, k'}^{i'}$ wp. $\Omega(\frac{q_S^2}{2^{\seedsize}})$ and with constant probability, this happens wp. $c^i_k = 2$ and $c^{i'}_{k'} \neq 2$. So we proved that the probability that the adversary finds a pair of couples $((i,k),(i',k'))$ satisfying Equation~\ref{Eq:AttackCondition}, which proves our theorem.
\end{proof}

As a concrete example, assume we consider the above signature scheme with $\lambda = 128$ which is practice the most common choice. This means that there is an adversary performing $\frac{1}{R}2^{64} \le 2^{58}$ queries and running in time $2^{64}$ that will break the scheme with constant probability (here something around $\frac{1}{10}$), while the underlying syndrome decoding problem requires time $2^{128}$ to break. The attack clearly invalidates the claim of $128$ bits of security. Notice also that it is quite standard to allow up to $2^{64}$ calls to the signing oracle, this is for example what is requried by the NIST for the standardisation of post-quantum signature schemes so our attack fits these guidelines. 

Fortunately, this attack can be mitigated with the methods we present in the next section with a close to zero loss in the efficiency, so optimized Stern's schemes can still be used securly with these mild changes.

\section{The $\salt$ + index construction}~\label{Section;FullScheme}
There are $2$ easy ways of mitigating the above attack: by taking $\seedsize = 2\lambda$ or by using randomized commitments but both these solutions are very costly in terms of signature size. Here, we present another solution that will very mildly increase the signature size while preserving security: the $\salt$ + index construction. The idea is very simple: when running the identification scheme (the full $R$ round scheme), we pick a random $\salt \in \zo^{2\lambda}$ which we use as a randomization for all the commitments and the pseudo-random generators. Since the functions are used several times, we add an index which will be like an internal counter. Each time such a function is used, we increment the index and add it to the input, to ensure that each function used is independent from one another. 

\begin{openbox}[\textit{Stern's identification scheme (optimized, $R$ rounds), $\IS_{Stern,opt}$}]
	\textbf{Initialization.} Let $(\Hm,\ev,\sv = \Hm\ev) \xleftarrow[]{\$} \mathcal{D}_{n,k,w}$ \ref{Conjecture:AvgHardnessBSD}, $sk = \ev$, $pk = (\Hm,\sv)$. The prover obtains $(pk,sk)$ and the verifier obtains $pk$.\bigskip
	
	\textbf{Interaction.} 
	The prover picks a random $\salt \in \zo^{2\lambda}$ and initializes $index = 0$. Each time index is called, increment index.
	\begin{itemize}
		\item  For $j$ from $1$ to $R$: the prover picks a random $\seed^j \in \zo^{\seedsize}$ and lets $(\seed^j_{\pi},\seed^j_{\yv}) = E_0(\seed^j,\salt,index), \pi^j = E_1(\seed^j_\pi,\salt,index)$ and \\ $\yv^j = (\pi^j)^{-1}(E_2(\seed^j_{\yv}),\salt,index)$. He then computes $\tv^j = \Hm  \yv^j$, $z^j_1 = (\pi^j,\tv^i),z_2 = \pi^i(\yv^i), z^i_3 = \pi^i(\yv^i + \ev)$, and computes for $i \in \{1,2,3\}, \ x^j_i = \hh_{comm}(z^j_i,\salt,index)$.
	\end{itemize}
	The prover sends $\widetilde{x} = \hh_{\comm}(x_1^1,x_2^1,x_3^1,\dots,x_1^R,x_2^R,x_3^R,\salt,index)$ to the verifier as well as $\salt$.
	\begin{itemize}
		\item For $j$ from $1$ to $R$: the verifier sends a random challenge $c^j \xleftarrow[]{\$} \{1,2,3\}$ to the prover.\\
		
		\item For $j$ from $1$ to $R$: the prover reveals information so that the verifier can recover $x^j_i$ for each $i \in \{1,2,3\} \setminus c$. More precisely:
		\begin{enumerate}
			\item if $c^j = 1$, the prover sends $l^j_1 = (o^j_2,o^j_3)$ with $o^j_2 = \seed_{\yv}, \ o^j_3 = \pi^j(\ev)$ and $x^j_1$.
			\item if $c^j = 2$, the prover sends $l^j_2 = (o^j_1,o^j_3)$ with $o^j_1 = \seed_{\pi}, \ o^j_3 = \yv^j + \ev$ and $x^j_2$
			\item if $c^j = 3$, the prover sends $l^j_3 = \seed$ and $x_3^j$.
		\end{enumerate}
	\end{itemize}	
	\textbf{Verification.} For $j$ from $1$ to $R$: 
	\begin{enumerate}
		\item if $c^j =1$, recovers $z^j_2,z^j_3$ from $o^j_2,o^j_3$ by taking $z^j_2 = E_2(o^j_2,\salt,index)$ and $z^j_3 = z^j_2 + o^j_3$. Then for $i \in \{2,3\}$ checks that $\hh_{comm}(x^j_i,\salt,index) = z^j_i$ and that $wt(o^j_3) = w$.
		\item if $c^j = 2$, computes $ \pi^j = E_1(o^j_1,\salt,index), z^j_3 = \pi^j(o^j_3)$ then $\tv^j = \Hm o^j_3 - \sv$ and finally $z^j_1 = (\pi^j,\tv^j)$. Then for $i \in \{1,3\}$ checks that $\hh_{comm}(x^j_i,\salt,index) = z^j_i$.
		\item if $c^j = 3$, computes $(\seed^j_{\pi},\seed^j_y) = E_0(l^j_{3},\salt,index)$,  $\pi^j = E_1(\seed^j_\pi,\salt,index)$ and $z^j_2 =E_2(\seed^j_{\yv},\salt,index)$, followed by $\tv^j = \Hm (\pi^j)^{-1}(\yv^j)$. From there, he constructs $z^j_1 = (\pi^j,\tv^j)$.  Then for $i \in \{1,2\}$ checks that $\hh_{comm}(x^j_i,\salt,index) = z^j_i$.
	\end{enumerate}
	The verifier also checks that $\hh_{\comm}(x_1^1,x_2^1,x_3^1,\dots,x_1^R,x_2^R,x_3^R,\salt,index)$ is equal to the value $\widetilde{x}$ he received. In each function, is known what index is used by the prover (it doesn't depend on some internal randomness) so the verifier uses for each function the same index as the prover when performing the above.
\end{openbox}

First, notice that this mitigates the above attack. Indeed, the attack used the fact that collision in the seeds between different signatures led to collisions in the commitments which could be detected and used. Here, with the addition of the $\salt$, this won't happen anymore (and the index ensures that collisions within the same signature cannot be used). Moreover, notice that in terms of communication, a single string $\salt$ of size $2 \lambda$ is added so the increase in the signature size is $2\lambda$.

There could a priori be other attacks on this scheme. We show that this is not the case, by showing that Stern's scheme with the $\salt$ + index construction is multi-HVZK. Then, one can use the results of~\cite{GHHM21} to prove the EUF-CMA security of the resulting signature scheme.

\subsection{Zero-knowledge for $\IS_{Stern}$ without pseudo-random generators}\label{Section:multiHVZK}
We first prove a bound on the multi-HVZK advantage when we don't have seeds, but we still have deterministic commitments. We call this scheme $\IS_{DC,NoS,\salt}$, which is described as follows: \\ \\
\cadre{$\IS_{DC,NoS,\salt}$ scheme \\ \hline 
	\textbf{Initialization.} Let $(\Hm,\ev,\sv = \Hm\ev) \xleftarrow[]{\$} \mathcal{D}_{n,k,w}$ \ref{Conjecture:AvgHardnessBSD}, $sk = \ev$, $pk = (\Hm,\sv)$. The prover obtains $(pk,sk)$ and the verifier obtains $pk$.\bigskip
	
	\textbf{Interaction.} 
	\begin{itemize}
		\item  For $j$ from $1$ to $R$: The prover picks a random permutation $\pi^j$ acting on $\F_q^n$ and a random column vector $\yv^j \xleftarrow[]{\$} \F_q^n$, and then computes $\tv^j = \Hm  \yv^j$. The prover computes $z^j_1 = (\pi^j,\tv^j)$, $z^j_2 = (\pi^j(\yv^j))$, $z_3 = (\pi^j(\yv^j + \ev))$, and computes for $i \in \{1,2,3\}, \ x^j_i = \hh_{comm}(z^j_i||\salt,index)$.
	\end{itemize}
	The prover sends $\widetilde{x} = \hh_{\comm}(x_1^1,x_2^1,x_3^1,\dots,x_1^R,x_2^R,x_3^R,\salt,index)$ to the verifier.
	\begin{itemize}
		\item For $j$ from $1$ to $R$: the verifier sends a random challenge $c^j \xleftarrow[]{\$} \{1,2,3\}$ to the prover.
		
		\item For $j$ from $1$ to $R$: the prover reveals information so that the verifier can recover $x^j_i$ for each $i \in \{1,2,3\} \setminus c$. More precisely:
		\begin{enumerate}
			\item if $c^j = 1$, the prover sends $l^j_1 = (o^j_2,o^j_3)$ with $o^j_2 = \pi^j(\yv^j), \ o^j_3 = \pi^j(\ev)$ and $x^j_1$.
			\item if $c^j = 2$, the prover sends $l^j_2 = (o^j_1,o^j_3)$ with $o^j_1 = \pi^j, \ o^j_3 = \yv^j + \ev$ and $x^j_2$
			\item if $c^j = 3$, the prover sends $l^j_3 = (\pi^j,\yv^j)$ and $x_3^j$.
		\end{enumerate}
	\end{itemize}	
	\textbf{Verification.} For $j$ from $1$ to $R$: 
	\begin{enumerate}
		\item if $c^j =1$, he recovers $z^j_2,z^j_3$ from $l^j_1$. Then for $i \in \{2,3\}$ checks that $\hh_{comm}(x^j_i) = z^j_i$ and that $wt(o^j_3) = w$.
		\item if $c^j = 2$, he recover $z^j_1,z^j_3$ from $l^j_2$. Then for $i \in \{1,3\}$ checks that $\hh_{comm}(x^j_i) = z^j_i$.
		\item if $c^j = 3$, he recover $z^j_1,z^j_2$ from $l^j_3$. Then for $i \in \{1,2\}$ checks that $\hh_{comm}(x^j_i) = z^j_i$.
	\end{enumerate}
	The verifier also checks that $\hh_{\comm}(x_1^1,x_2^1,x_3^1,\dots,x_1^R,x_2^R,x_3^R,\salt,index)$ is equal to the value $\widetilde{x}$ he received. 
}

We prove a tight reduction from the advantage of breaking the $t$-HVZK condition and the advantage of solving the underlying syndrome decoding problem. 

\begin{theorem} ~\label{Theorem:4}
	For any adversary $\aa$, there exists a adversary $\bb$ st. 
	$$Adv^{\mathrm{t-HVZK}}_{\IS_{DC,NoS,\salt}}(\aa) \le Adv^{\SD}(\bb) + O\left(\frac{t^2}{2^{\saltsize}}\right).$$ 
	where $Adv^{\SD}(\bb)$ gives the probability that the adversary breaks the underlying syndrome decoding problem, and the running time of $\bb$ is the same as the running time of $\aa$.
\end{theorem}
Before presenting a proof of this theorem, let us perform a small discussion on how to interpret this theorem. In the previous section, we started from a signature scheme st. the underlying syndrome decoding problem requires $2^\lambda$ time to break, and we presented an adversary $\aa$ running in time $O(2^{\lambda/2})$ that breaks the signature scheme arising from $\IS_{Stern,opt}$ with $t = O(2^{\lambda/2})$ signature queries. One can also interpret our attack as an adversary $\aa$ st.  running in time $O(2^{\lambda/2})$ st. $Adv^{\mathrm{t-HVZK}}_{\IS_{Stern,opt}}(\aa)$ is close to $1$ with $t = 2^{\lambda/2}$. 

Theorem~\ref{Theorem:4} that we will shortly prove, shows that such an attack cannot be possible when in the case where there is salt. Indeed, assume we also had an adversary $\aa$ running in time $2^{\lambda}$ st. $Adv^{\mathrm{t-HVZK}}_{\IS_{Stern,opt}}(\aa)$ is close to $1$. Then the above theorem immediately says that there is an adversary $\bb$ running in time $2^{\lambda/2}$ that breaks the underlying syndrome decoding problem but we chose the parameters st. such an adversary would require time at least $2^{\lambda}$ hence the contradiction.

We now present the proof of this theorem. 
\begin{proof}
	We will do a game based proof. The first game we consider is 
\begin{center}
	\begin{minipage}{0.99\textwidth}
		\cadre{\vspace{0.1cm} \game 0 Real transcript \\
			\hline
			$(pk,sk) \leftarrow \ISKeygen$. \\
			$\forall i \in \{1,\dots,t\}:$ 
			\begin{itemize}
				\item Pick a random $\salt^i \Unif \zo^{\saltsize}$.
				\item $\forall j \in \{1,\dots,R\}$:
			\begin{itemize}
				\item Pick a random permutation $\pi^{ij}$ and a random string $y^{ij} \in \F_q^n$. Let $\tv^{ij} = \Hm \yv^{ij}$. Let $z^{ij}_1 = (\pi^{ij},\tv^{ij})$, $z^{ij}_2 = (\pi^{ij}(\yv^{ij}))$, $z^{ij}_3 = (\pi^{ij}(\yv^{ij} + \ev))$, and compute for $k \in \{1,2,3\}, \ x^{ij}_k = \hh_{comm}(z^{ij}_k||\salt^i,index)$. Let also 
				$$ \widetilde{x}^i = \hh_{\comm}(x_1^{i1},x_2^{i1},x_3^{i1},\dots,x_1^{iR},x_2^{iR},x_3^{iR},\salt^i,index).$$
				Pick a random $c^{ij} \Unif \{1,2,3\}$ 
				\begin{enumerate}
					\item If $c^{ij} = 1$, set $l^{ij}_1 = (o^{ij}_2,o^{ij}_3)$ with $o^{ij}_2 = \pi^{ij}(\yv^{ij}), \ o^{ij}_3 = \pi^{ij}(\ev)$.
					\item If $c^{ij} = 2$, set $l^{ij}_2 = (o^{ij}_1,o^{ij}_3)$ with $o^{ij}_1 = \pi^{ij}, \ o^{ij}_3 = \yv^{ij} + \ev$.
					\item If $c ^{ij}= 3$, set $l^{ij}_3 = (o^{ij}_1,o^{ij}_2)$ with $o^{ij}_1 = \pi^{ij}, \ o^{ij}_2 = \yv^{ij}$.
				\end{enumerate} 
					\item $\tau^i = \left(\widetilde{x}^i,\salt^i,x_{c^{i1}}^{i1},\dots,x_{c^{iR}}^{iR},l_{c^{i1}}^{i1},\dots,l_{c^{iR}}^{iR}\right).$
			\end{itemize}

			\end{itemize}
			$\aa$ wins the game iff. $\aa(\tau^{1},\dots,\tau^{t},pk) = 1$.}
	\end{minipage} 
\end{center}

In this game, we construct $t$ transcripts constructed according to $trans_{\IS_{Stern,DC,NoS}}(pk,sk)$.  The game is won if the adversary outputs $1$ given these real transcripts as input. 

We now consider another game where we construct the same transcripts but in a different way. For each choice of $c^{ij}$, it is  possible to construct the responses $l^{ij}$ independently of the secret key and to make only the strings $x^{ij}_{c^{ij}}$ dependent of the secret key. Let $S_w$ be set the of words of weight $w$.

\begin{center}
	\begin{minipage}{0.99\textwidth}
		\cadre{\vspace{0.1cm} \game 1 Real transcript with openings depending only on $pk$ \\
			\hline
			$(pk,sk) \leftarrow \ISKeygen$. \\
			$\forall i \in \{1,\dots,t\}:$ 
			\begin{itemize}
				\item Pick a random $\salt^i \Unif \zo^{\saltsize}$.
				\item $\forall j \in \{1,\dots,R\}$:
				\begin{itemize}
					\item Pick a random challenge $c^{ij} \Unif \{1,2,3\}$.
				\begin{enumerate}
					\item If $c^{ij} = 1$, set $l^{ij}_1 = (o^{ij}_2,o^{ij}_3)$ with $o^{ij}_2 \Unif \F_q^n$ and $ o^{ij}_3 \Unif S_w$.  Pick a random $\pi^{ij}$ st. $\pi^{ij}(o_3^{ij}) = \ev$, and let $x^{ij}_1 = \hh_{comm}\left(\pi^{ij},\tv^{ij} = \Hm \pi^{ij}(o_2^{ij}),\seed,index\right)$.
					\item If $c^{ij} = 2$, set $l^{ij}_2 = (o^{ij}_1,o^{ij}_3)$ with $o^{ij}_1 \Unif \Perm_n, \ o^{ij}_3 \Unif \F_q^n$. Pick $x_2^{ij} = \hh_{comm}\left(\pi^{ij}(\yv^{ij}),\salt^i,index\right)$.
					\item If $c ^{ij}= 3$, set $l^{ij}_3 = (o^{ij}_1,o^{ij}_2)$ with $o^{ij}_1 \Unif \Perm_n, \ o^{ij}_2 \Unif \F_q^n$. Pick $x_3^{ij} = \hh_{comm}\left(\pi^{ij}(\yv^{ij} + \ev),\salt^i,index\right)$.
					\end{enumerate}
 Let also 
					$$ \widetilde{x}^i = \hh_{\comm}(x_1^{i1},x_2^{i1},x_3^{i1},\dots,x_1^{iR},x_2^{iR},x_3^{iR},\salt^i,index).$$

					\item $\tau^i = \left(\widetilde{x}^i,\salt^i,x_{c^{i1}}^{i1},\dots,x_{c^{iR}}^{iR},l_{c^{i1}}^{i1},\dots,l_{c^{iR}}^{iR}\right).$
				\end{itemize}
			\end{itemize}
			$\aa$ wins the game iff. $\aa(\tau^{1},\dots,\tau^{t},pk) = 1$.}
	\end{minipage} 
\end{center}

\begin{lemma}
	$$\Pr[\aa \textrm{ wins } \game 0] = \Pr[\aa \textrm{ wins } \game 1].$$
\end{lemma}
\begin{proof}
	This is true since the distribution of each $\tau^{i}$ is the same for both games. This can be seen using the same arguments as in Proposition~\ref{Proposition:1roundHVZK}, where we proved that the ideal simulator has the same distribution of openings as the real transcript. 
\end{proof}

We now change $\game 1$ into $\game2$ where we replace the commitments $x^{ij}$ with uniformly random strings. We put the difference between the games in red.

\begin{center}
	\begin{minipage}{0.99\textwidth}
		\cadre{\vspace{0.1cm} \game 2 Simulated transcript \\
			\hline
			$(pk,sk) \leftarrow \ISKeygen$. \\
			$\forall i \in \{1,\dots,t\}:$ 
			\begin{itemize}
				\item Pick a random $\salt^i \Unif \zo^{\saltsize}$.
				\item $\forall j \in \{1,\dots,R\}$:
				\begin{itemize}
					\item Pick a random challenge $c^{ij} \Unif \{1,2,3\}$.
					\begin{enumerate}
						\item If $c^{ij} = 1$, set $l^{ij}_1 = (o^{ij}_2,o^{ij}_3)$ with $o^{ij}_2 \Unif \F_q^n$ and $ o^{ij}_3 \Unif S_w$. 
						\item If $c^{ij} = 2$, set $l^{ij}_2 = (o^{ij}_1,o^{ij}_3)$ with $o^{ij}_1 \Unif \Perm_n, \ o^{ij}_3 \Unif \F_q^n$.
						\item If $c ^{ij}= 3$, set $l^{ij}_3 = (o^{ij}_1,o^{ij}_2)$ with $o^{ij}_1 \Unif \Perm_n, \ o^{ij}_2 \Unif \F_q^n$. 
					\end{enumerate}
					\item \textcolor{red}{Let $x^{ij}_{c^{ij}} \Unif \zo^{\commsize}$.}
					Let also 
					$$ \widetilde{x}^i = \hh_{\comm}(x_1^{i1},x_2^{i1},x_3^{i1},\dots,x_1^{iR},x_2^{iR},x_3^{iR},\salt^i,index).$$
					
					\item $\tau^i = \left(\widetilde{x}^i,\salt^i,x_{c^{i1}}^{i1},\dots,x_{c^{iR}}^{iR},l_{c^{i1}}^{i1},\dots,l_{c^{iR}}^{iR}\right).$
				\end{itemize}
			\end{itemize}
			$\aa$ wins the game iff. $\aa(\tau^{1},\dots,\tau^{t},pk) = 1$.}
	\end{minipage} 
\end{center}

\begin{lemma} We have
	$$ \left|\Pr[\aa \textrm{ wins } \game 2] - \Pr[\aa \textrm{ wins } \game 1]\right| \le Adv^{\SD}(\bb) + \Theta(\frac{t^2}{2^{\saltsize}}).$$
	where $\bb$ is an adversary st. $|\bb| \approx |\aa|$ and $Adv^{\SD}(\bb) $ is the probability that $\bb$ breaks the underlying syndrome decoding problem.
\end{lemma}
\begin{proof}
	Consider an adversary $\aa$ for $\game 2$. This adversary can perform queries $\hh_{\comm}(\cdot)$. The only difference between $\game 1$ and $\game 2$ is in the strings $x^{ij}_{c^{ij}}$. In $\game 1$, if all the salts $\seed_{salt}^i$ are pairwise distinct,  $x^{ij}_{c^{ij}} = \hh_{\comm}(z^{ij}_{c^{ij}},\salt^i,index)$ is queried only once to while $\hh_{\comm}(*,\salt^i,index)$  is not queried in $\game 2$. The probability that all the salts are different is $\Theta(\frac{t^2}{2^{\saltsize}})$.
	
	 $\aa$ can query this function and the only way the output distribution of $\aa$ is different in $\game 2$ and in $\game 1$ is $\aa$ queries exactly $z^{ij}_{c^{ij}}$.
	 This means 
	$$ \left|\Pr[\aa \textrm{ wins } \game 2] - \Pr[\aa \textrm{ wins } \game 1]\right| \leq \Pr[\exists i,j \ st. \ \aa \textrm{ queries } x^{ij}_{c^{ij}}] + \Theta(\frac{t^2}{2^{\saltsize}}).$$
	
	Moreover, from the openings $o^{ij}$, the adversary can recover the $2$ other strings $z^{ij}_{c'},z^{ij}_{c''}$ for $c' \neq c''$ and $c',c'' \neq c^{ij}$. From the $3$ strings $z^{ij}_1, \ z^{ij}_2, z^{ij}_3$ one can recover in polynomial time the secret key $sk$ using some extraction algorithm (we described how to do this in Section~\ref{Section:Attack} for example).  We can therefore construct from $\aa$ an algorithm $\bb$ that solves the syndrome decoding problem. \\
	
	Algorithm $\bb$: run $\aa$. Each time a query $\hh_{\comm}(\cdot)$ is done (and the values $\seed_{salt}^i$ and $c^{ij}$ are publicly known from the transcript), consider the responses $z^{ij}$ from the transcript. From them, one can recover the values $x^{ij}_{c'},x^{ij}_{c''}$ for $c' \neq c''$ and $c',c'' \neq c^{ij}$, and check, using the extractor $Ext$ and by arranging the inputs depending on the value of $c^{ij}$, whether from $Ext(x^{ij}_{c'},x^{ij}_{c''})$ one can recover the secret key $sk$. Since $Ext$ is efficient, the running time of $\bb$ is essentially the running time of $\aa$. This means that 
	$$\Pr[\exists i,j \ st. \ \aa \textrm{ queries } x^{ij}_{c^{ij}}] \le Adv^{\SD}(\bb)$$
	with $|\bb| \approx |\aa|$.\\

Putting the $2$ inequalities together, we have 
	$$ \left|\Pr[\aa \textrm{ wins } \game 2] - \Pr[\aa \textrm{ wins } \game 1]\right| \le Adv^{\SD}(|\aa|) + \Theta(\frac{t^2}{2^{\saltsize}}).$$
\end{proof}

Finally, notice that the $\tau^{ij}$ constructed in $\game 2$ follows exactly the distribution of simulated transcripts.  This means 

\begin{align*}
	Adv^{\mathrm{t-HVZK}}_{\IS_{DC,NoS,\salt}}(\aa)  = \left|\Pr[\aa \textrm{ wins } \game 2] - \Pr[\aa \textrm{ wins } \game 0]\right|.
\end{align*}
Putting everything together, we obtain 
$$Adv^{\mathrm{t-HVZK}}_{\IS_{DC,NoS,\salt}}(\aa) \le Adv^{\SD}(|A|) + O\left(\frac{t^2}{2^{\saltsize}}\right).$$ 
\end{proof}

\subsection{Zero-knowledge for $\IS_{Stern}$ using pseudo-random generators}
Here, we show that using pseudo-random strings instead of random strings $\pi,\yv$ doesn't change the security of the scheme. Recall that we showed in the previous section that if one isn't careful using seeds of size $\lambda$ and deterministic commitments, there are attacks on the identification scheme. We show here that our $\salt$ + $index$ construction. In the ROM, we have the following:
\begin{proposition}\label{Proposition:SeedsPartial}
	Consider a random function $E : \zo^{\seedsize} \rightarrow I$ modeled  as a random function in the random oracle model, and consider the following distinguishing game:
	\begin{center}
		\begin{minipage}{0.96\textwidth}
			\cadre{
				$\game Indist(E)$ \\
				\hline 
				$b \Unif \zo$. \\
				If $b = 0$, $\seed \Unif \zo^{\seedsize}, y = E(\seed)$. \\
				If $b = 1$, $y \Unif I$. \\
				$b' \leftarrow A(y)$. \\
				The game is won iff. $b = b'$.
			}
		\end{minipage}
	\end{center}
	Consider an adversary $\aa$ that performs $q$ queries to $E$. Then 
	$$ \Pr[\aa \textrm{ wins } \game Indist(E)] \le \frac{1}{2} + O(\frac{q}{2^{\seedsize}}).$$
\end{proposition}
Moreover, the above proposition is true even if $\aa$ has some auxiliary information, \emph{as long as it is independent of the function $E$}\footnote{This last part is what motivates the use of the salt + index construction for seeds.}.

\begin{proposition}
	Consider the $2$ following games
	\begin{center}
		\begin{minipage}{0.99\textwidth}
			\cadre{\vspace{0.1cm} \game 0 Real transcript \\
				\hline
				$(pk,sk) \leftarrow \ISKeygen(1^\lambda)$. \\
				$\forall i \in \{1,\dots,t\}, \tau^1,\dots,\tau^t\leftarrow trans_{\IS_{DC,NoS,\salt}}(pk,sk)$. \\
				$\aa$ wins the game iff. $\aa(\tau^{1},\dots,\tau^{t}) = 1$.}
		\end{minipage} 
	\end{center}
	
	Notice that this game is exactly $\game 0$ from Section~\ref{Section:multiHVZK}.  Consider the second game
	
	\begin{center}
		\begin{minipage}{0.99\textwidth}
			\cadre{\vspace{0.1cm} \game 0' Real transcript with seeds\\
				\hline
				$(pk,sk) \leftarrow \ISKeygen$. \\
				$\forall i \in \{1,\dots,t\}, \tau^1,\dots,\tau^t\leftarrow trans_{\IS_{Stern,opt}}(sk)$. \\
				$\aa$ wins the game iff. $\aa(\tau^{1},\dots,\tau^{t}) = 1$.}
		\end{minipage} 
	\end{center}

where $\IS_{Stern,opt}$ is the full optimized scheme with seeds and salt described in Section~\ref{Section;FullScheme}.
	
	For any adversary performing $q$ queries to $E$, we have
	$$ \Pr[\aa \ { \wins } \ \game 0] - \Pr[\aa \ { \wins } \ \game 0'] \le O(\frac{q}{2^{\seedsize}}) + O(\frac{t^2}{2^{\saltsize}}).$$
\end{proposition}
\begin{proof}
	The only difference between the $2$ games is the way the strings $\pi^{ij}$ and $\pi^{ij}(\yv^{ij})$ are generated. Notice that for each $i,j$ we use a different function $E$ as long as the salts are all different (more precisely, we use $E(*||\salt^{i},index)$ so we actually use the same function but with a different suffix). Take the permutations $\pi^{ij}$, the adversary can distinguish a pseudo-random $\pi^{ij}$ from a random $\pi^{ij}$ wp. $O(\frac{q^{ij}}{2^{\seedsize}})$ where $q^{ij}$ is the number of queries done to the function $E$ used to construct $\pi^{ij}$. Similarly, let $r^{ij}$ be the number of queries done to the function $E$ used to construct $\pi^{ij}(\yv^{ij})$. Because each of these functions is different, we have $\sum_{ij} q_{ij} + r_{ij} \le q$ where recall that $q$ is the total number of queries done by $\aa$. Moreover, we can use Proposition~\ref{Proposition:SeedsPartial} and therefore have
	$$ \Pr[\aa \mbox{ wins } \game 0] - \Pr[\aa \mbox{ wins } \game 0'] \le \sum_{ij} O(\frac{q_{ij} + r_{ij}}{2^{\seedsize}}) +  O(\frac{t^2}{2^{\saltsize}}) \le O(\frac{q}{2^{\seedsize}}) + O(\frac{t^2}{2^{\saltsize}}).$$
\end{proof}

\section{Reduction of Signature Size using Hash Trees}
	
	A common technique for further reducing the signature size is to utilize a technique based upon seeded hash trees -- small seed expansions associated to Merkel trees. This technique is used, for example, in \cite{Beu20} and promises to have a huge advantage in efficiency. We show here how this technique can reduce the length of the signatures arising from Stern's signature scheme.
	
	\begin{openbox}[\textit{Stern's identification protocol, optimized fusion of $4$ rounds}]
		 \textbf{Initialization.} \textbf{Keygenerator} samples $(\Hm,\ev,\sv = \Hm\ev) \xleftarrow[]{\$} \mathcal{D}_{n,k,w}$ \ref{Conjecture:AvgHardnessBSD}, $\textbf{sk} = \ev$, $\textbf{pk} = (\Hm,\sv)$. \\
		
		 \textbf{Interaction.}  \\
		
		 \textbf{Prover}: picks $\seed_1^{1234},\seed_2^{1234} \in \zo^{\seedsize}$. For $j \in \{1,2\}$: 
		\begin{align*}
			(\seed_j^{12},\seed_j^{34}) & = J(\seed_j^{1234}|| \salt || index) \\
			(\seed_j^{1},\seed_j^2) & = J(\seed_j^{1234}|| \salt || index) \\
			(\seed_j^{1},\seed_j^2) & = J(\seed_j^{1234}|| \salt || index) 
		\end{align*}
		For $i \in \{1,2,3,4\}$, we construct $\pi^i = E(\seed_1^i || \salt || index )$ and  $\pi^i(\yv^i) = E(\seed_1^i || \salt || index )$. Let $\tv^i = \Hm \yv^i$, $\zv^i_1 = (\pi^i,\tv^i)$. He then construct the commitments $\xv^i_j = \hh_{\comm}(z^i_j||\salt||index)$, and constructs 
		\begin{align*}
			\xv^{12}_j & = L((\xv^1_j,\xv^2_j)||\salt||index)) \\
			\xv^{34}_j & = L((\xv^3_j,\xv^4_j)||\salt||index))  \\
			\xv^{1234}_j & = L((\xv^{12}_j,\xv^{34}_j)||\salt||index)) 
		\end{align*} 
		
		\textbf{Prover} calculates random \textit{challenges} as $( c^1,c^2,c^3,c^4) = \hh_{FSH}(\xv^{1234},m).$\\
		
		The prover and the verifier run a $4$ round optimized scheme with the following optimizations:
		\begin{itemize}
			\item For $j \in \{1,2\}$ If the prover needs to reveal $\seed_j^1$ and $\seed_j^2$ (resp. $\seed_j^3$ and $\seed_j^4$) he reveals $\seed_j^{12}$ instead (resp. $\seed_j^{34}$) and the verifier uses $J$ to recover each seed. If the prover needs to reveal $\seed_j^{12}$ and $\seed_j^{34}$ he reveals $\seed_j^{1234}$ instead and uses $J$ again
			\item For $j \in \{1,2,3\}$ If the prover needs to reveal $x_j^1$ and $x_j^2$ (resp. $x_j^3$ and $X_j^4$) he reveals $x_j^{12}$ instead (resp. $x_j^{34}$) and the verifier uses $L$ to recover each commitment. If the prover needs to reveal $x_j^{12}$ and $x_j^{34}$ he reveals $x_j^{1234}$ instead and uses $L$ again.
		\end{itemize}
	\end{openbox}
	
	Stern's identification scheme with hash trees consists of $55$ parallel repetition of the above scheme (which each has $4$ rounds), except that the same salt is used for all the repetition. This gives a total of $220$ rounds. 
	\begin{proposition}
		The signature size of the signature scheme constructed from the above identification scheme is (in bits):
		$$  SIZE = \saltsize + 55\left(2 C_1 + 4\left(\frac{1}{3}s_w + \frac{1}{3}n\log_2(q) + \frac{1}{3} \commsize\right)\right).$$
		where 
		$$ C_1 = \frac{128}{81}\seedsize + \frac{89}{81}\commsize.$$
		and $s_w = \log_2 \left(|\{x \in \F_q^n, wt(x) = w\}\right)$.
	\end{proposition}

\begin{proof}
 Let us compute the communication cost in this setting. Let $C_1$ be the communication cost corresponding to $j = 1$. This means we count here the costs of the revealed $\seed^i_1$ as well as $x_1^i$. For each of the $4$ rounds, we have a probability $\frac{2}{3}$ of revealing the seed and probability $\frac{1}{3}$ of revealing the commitment. Moreover, we use our seed and commitment regrouping. For example, if we reveal the $4$ seeds (which happens wp. $(\frac{2}{3})^4)$, we only have to reveal $\seed_1^{1234}$ which costs $\seedsize$ bits. Considering each of these case, we obtain:
	
	\begin{multline*}C_1 = (\frac{2}{3})^4 s_{\seed} + 4(\frac{2}{3})^3\frac{1}{3} \left(2s_{\seed} + s_{\comm}\right) + 4(\frac{2}{3})^2(\frac{1}{3})^{2}\left(2s_{\seed} + 2s_{\comm}\right) \\ + 2(\frac{2}{3})^2(\frac{1}{3})^{2}\left(s_{\seed} + s_{\comm}\right) + 4\frac{2}{3}(\frac{1}{3})^3 \left(s_{\seed} + 2s_{\comm}\right) +  (\frac{1}{3})^4 s_{\comm} 
	\end{multline*}
which gives 
	$$ C_1 = \frac{128}{81}s_{\seed} + \frac{89}{81}s_{\comm}.$$
	The same calculation can be done for $j = 2$. For the final case, we reveal the commitment wp. $\frac{1}{3}$. If this doesn't happen, we reveal $\pi(e)$ or $\pi(y+e)$ which requires respectively $s_w$ and $n\log_2(q)$ bits to send. Finally, we have to send also the salt used (the indexes can be recomputed by the verifier). From there, we have that the total size is 
	$$  SIZE = \saltsize + 55\left(2 C_1 + 4\left(\frac{1}{3}s_w + \frac{1}{3}n\log_2(q) + \frac{1}{3} \commsize\right)\right).$$
\end{proof}
	
\subsection{Proposal with Lee's metric}
We use our scheme with the salt + index construction and with the above hash trees. Wa also consider Lee's metric, defined as follows. For $\ev = (e_1,\dots,e_n) \in F_q^n$, we define
$$ wt_L(\ev)= \sum_{i  = 1}^n \min\{e_i,q-e_i\}.$$
The best known algorithms for solving the syndrome decoding problem with this metric were proposed in~\cite{CDE21}. Plugging these exponents to our construction, we obtain the following signature sizes.

\begin{center}
	\begin{tabular}{ |c||c|c|c|c|c|  }
		\hline
		\multicolumn{6}{|c|}{Non-optimized version} \\
		\hline
		$q$& $2$& $3$& $5$& $7$& $13$\\
		\hline
		$wt_H(\cdot)$& $253.05$ kB & $116.54$ kB & $138.54$ kB & $126.468$ kB & $113.227$ kB\\
		$wt_L(\cdot)$& $253.05$ kB & $116.54$ kB & $95.48$ kB & $90.935$ kB & $79.2733$ kB\\
		\hline
	\end{tabular}
\end{center}

\begin{center}
	\begin{tabular}{ |c||c|c|c|c|c|  }
		\hline
		\multicolumn{6}{|c|}{Optimized version} \\
		\hline
		$q$& $2$& $3$& $5$& $7$& $13$\\
		\hline
		$wt_H(\cdot)$& $26.2105$ kB & $21.8122$ kB & $27.6249$ kB & $28.291$ kB & $29.3769$ kB\\
		$wt_L(\cdot)$& $26.2105$ kB & $21.8122$ kB & $21.4108$ kB & $22.7065$ kB & $23.2905$ kB\\
		\hline
	\end{tabular}
\end{center}

\begin{center}
	\begin{tabular}{ |c||c|c|c|c|c|  }
		\hline
		\multicolumn{6}{|c|}{Hash-tree version (4 levels)} \\
		\hline
		$q$& $2$ & $3$& $5$& $7$& $13$\\
		\hline
		$wt_H(\cdot)$& $24.7223$ kB & $20.304$ kB & $26.1432$ kB & $26.8123$ kB & $27.9032$ kB\\
		$wt_L(\cdot)$& $24.7223$ kB & $20.304$ kB & $19.9007$ kB & $21.2023$ kB & $21.789$ kB\\
		\hline
	\end{tabular}
\end{center}
\section{Conclusion}
In this work, we studied the security and insecurity of some optimized Stern's signature schemes. The main takeaway from our work is that one has to be careful when considering deterministic commitments and seeds of size $\lambda$ but that there is a way to make work with minimal deterioration of the efficiency. Similar analysis could be done for other schemes, and we hope to extend this fully to MQDSS. There doesn't seem to be any special difficulty but each scheme has to be studied with care. Moreover, it would be nice to extend this framework to $5$ round schemes, which is an optimization that we didn't consider in this work. 
\bibliography{paper2}
\bibliographystyle{alpha}
\end{document}